\documentclass[journal]{IEEEtran}

\usepackage[utf8]{inputenc}

\usepackage{amsmath}
\usepackage{amssymb}
\usepackage{subfigure}
\usepackage{graphicx}
\usepackage{graphics}
\usepackage{color}
\usepackage{psfrag}
\usepackage{cite}
\usepackage{balance}
\usepackage{algorithm}
\usepackage{accents}
\usepackage{amsthm}
\usepackage{bm}
\usepackage{url}
\usepackage{algorithmic}
\usepackage[english]{babel}
\usepackage{multirow}
\usepackage{enumerate}
\usepackage{cases}
\usepackage{stfloats}
\usepackage{dsfont}
\usepackage{soul}
\usepackage{amsfonts}
\usepackage{fancyhdr}
\usepackage{hhline}
\usepackage{array}
\usepackage{booktabs}
\usepackage{framed}
\usepackage{float}
\usepackage{soul}
\usepackage{url}

\begin{document}
\newtheorem{theorem}{Theorem}
\newtheorem{acknowledgement}[theorem]{Acknowledgement}
\newtheorem{axiom}[theorem]{Axiom}
\newtheorem{case}[theorem]{Case}
\newtheorem{claim}[theorem]{Claim}
\newtheorem{conclusion}[theorem]{Conclusion}
\newtheorem{condition}[theorem]{Condition}
\newtheorem{conjecture}[theorem]{Conjecture}
\newtheorem{criterion}[theorem]{Criterion}
\newtheorem{definition}{Definition}
\newtheorem{exercise}[theorem]{Exercise}
\newtheorem{lemma}{Lemma}
\newtheorem{corollary}{Corollary}
\newtheorem{notation}[theorem]{Notation}
\newtheorem{problem}[theorem]{Problem}
\newtheorem{proposition}{Proposition}
\newtheorem{solution}[theorem]{Solution}
\newtheorem{summary}[theorem]{Summary}
\newtheorem{assumption}{Assumption}
\newtheorem{example}{\bf Example}
\newtheorem{remark}{\bf Remark}
\newtheorem{design}{\bf Design}
\newtheorem{module}{\bf Module}
\newtheorem{circuit}{\bf Circuit}

\newtheorem{thm}{Corollary}[section]
\renewcommand{\thethm}{\arabic{section}.\arabic{thm}}

\def\qed{$\Box$}
\def\QED{\mbox{\phantom{m}}\nolinebreak\hfill$\,\Box$}
\def\proof{\noindent{\emph{Proof:} }}
\def\poof{\noindent{\emph{Sketch of Proof:} }}
\def
\endproof{\hspace*{\fill}~\qed
\par
\endtrivlist\unskip}
\def\endproof{\hspace*{\fill}~\qed\par\endtrivlist\vskip3pt}

\def\E{\mathsf{E}}
\def\eps{\varepsilon}
\def\phi{\varphi}
\def\Lsp{{\boldsymbol L}}
\def\Bsp{{\boldsymbol B}}
\def\lsp{{\boldsymbol\ell}}
\def\Ltsp{{\Lsp^2}}
\def\Lpsp{{\Lsp^p}}
\def\Linsp{{\Lsp^{\infty}}}
\def\LtR{{\Lsp^2(\Rst)}}
\def\ltZ{{\lsp^2(\Zst)}}
\def\ltsp{{\lsp^2}}
\def\ltZt{{\lsp^2(\Zst^{2})}}
\def\ninN{{n{\in}\Nst}}
\def\oh{{\frac{1}{2}}}
\def\grass{{\cal G}}
\def\ord{{\cal O}}
\def\dist{{d_G}}
\def\conj#1{{\overline#1}}
\def\ntoinf{{n \rightarrow \infty}}
\def\toinf{{\rightarrow \infty}}
\def\tozero{{\rightarrow 0}}
\def\trace{{\operatorname{trace}}}
\def\ord{{\cal O}}
\def\UU{{\cal U}}
\def\rank{{\operatorname{rank}}}
\def\acos{{\operatorname{acos}}}

\def\SINR{\mathsf{SINR}}
\def\SNR{\mathsf{SNR}}
\def\SIR{\mathsf{SIR}}
\def\tSIR{\widetilde{\mathsf{SIR}}}
\def\Ei{\mathsf{Ei}}
\def\l{\left}
\def\r{\right}
\def\lb{\left\{}
\def\rb{\right\}}

\def\pb{\pmb{[}}
\def\pk{\pmb{]}}

\setcounter{page}{1}

\newcommand{\eref}[1]{(\ref{#1})}
\newcommand{\fig}[1]{Fig.\ \ref{#1}}

\def\bydef{:=}
\def\ba{{\mathbf{a}}}
\def\bb{{\mathbf{b}}}
\def\bc{{\mathbf{c}}}
\def\bd{{\mathbf{d}}}
\def\bee{{\mathbf{e}}}
\def\bff{{\mathbf{f}}}
\def\bg{{\mathbf{g}}}
\def\bh{{\mathbf{h}}}
\def\bi{{\mathbf{i}}}
\def\bj{{\mathbf{j}}}
\def\bk{{\mathbf{k}}}
\def\bl{{\mathbf{l}}}
\def\bm{{\mathbf{m}}}
\def\bn{{\mathbf{n}}}
\def\bo{{\mathbf{o}}}
\def\bp{{\mathbf{p}}}
\def\bq{{\mathbf{q}}}
\def\br{{\mathbf{r}}}
\def\bs{{\mathbf{s}}}
\def\bt{{\mathbf{t}}}
\def\bu{{\mathbf{u}}}
\def\bv{{\mathbf{v}}}
\def\bw{{\mathbf{w}}}
\def\bx{{\mathbf{x}}}
\def\by{{\mathbf{y}}}
\def\bz{{\mathbf{z}}}
\def\b0{{\mathbf{0}}}

\def\bA{{\mathbf{A}}}
\def\bB{{\mathbf{B}}}
\def\bC{{\mathbf{C}}}
\def\bD{{\mathbf{D}}}
\def\bE{{\mathbf{E}}}
\def\bF{{\mathbf{F}}}
\def\bG{{\mathbf{G}}}
\def\bH{{\mathbf{H}}}
\def\bI{{\mathbf{I}}}
\def\bJ{{\mathbf{J}}}
\def\bK{{\mathbf{K}}}
\def\bL{{\mathbf{L}}}
\def\bM{{\mathbf{M}}}
\def\bN{{\mathbf{N}}}
\def\bO{{\mathbf{O}}}
\def\bP{{\mathbf{P}}}
\def\bQ{{\mathbf{Q}}}
\def\bR{{\mathbf{R}}}
\def\bS{{\mathbf{S}}}
\def\bT{{\mathbf{T}}}
\def\bU{{\mathbf{U}}}
\def\bV{{\mathbf{V}}}
\def\bW{{\mathbf{W}}}
\def\bX{{\mathbf{X}}}
\def\bY{{\mathbf{Y}}}
\def\bZ{{\mathbf{Z}}}

\def\bxi{{\boldsymbol{\xi}}}

\def\sT{{\mathsf{T}}}
\def\sH{{\mathsf{H}}}
\def\cmp{{\text{cmp}}}
\def\cmm{{\text{cmm}}}
\def\WPT{{\text{WPT}}}
\def\lo{{\text{lo}}}
\def\gl{{\text{gl}}}

\def\tT{{\widetilde{T}}}
\def\tF{{\widetilde{F}}}
\def\tP{{\widetilde{P}}}
\def\tG{{\widetilde{G}}}
\def\tbh{{\widetilde{\mathbf{h}}}}
\def\tbg{{\widetilde{\mathbf{g}}}}

\def\mA{{\mathbb{A}}}
\def\mB{{\mathbb{B}}}
\def\mC{{\mathbb{C}}}
\def\mD{{\mathbb{D}}}
\def\mE{{\mathbb{E}}}
\def\mF{{\mathbb{F}}}
\def\mG{{\mathbb{G}}}
\def\mH{{\mathbb{H}}}
\def\mI{{\mathbb{I}}}
\def\mJ{{\mathbb{J}}}
\def\mK{{\mathbb{K}}}
\def\mL{{\mathbb{L}}}
\def\mM{{\mathbb{M}}}
\def\mN{{\mathbb{N}}}
\def\mO{{\mathbb{O}}}
\def\mP{{\mathbb{P}}}
\def\mQ{{\mathbb{Q}}}
\def\mR{{\mathbb{R}}}
\def\mS{{\mathbb{S}}}
\def\mT{{\mathbb{T}}}
\def\mU{{\mathbb{U}}}
\def\mV{{\mathbb{V}}}
\def\mW{{\mathbb{W}}}
\def\mX{{\mathbb{X}}}
\def\mY{{\mathbb{Y}}}
\def\mZ{{\mathbb{Z}}}

\def\cA{\mathcal{A}}
\def\cB{\mathcal{B}}
\def\cC{\mathcal{C}}
\def\cD{\mathcal{D}}
\def\cE{\mathcal{E}}
\def\cF{\mathcal{F}}
\def\cG{\mathcal{G}}
\def\cH{\mathcal{H}}
\def\cI{\mathcal{I}}
\def\cJ{\mathcal{J}}
\def\cK{\mathcal{K}}
\def\cL{\mathcal{L}}
\def\cM{\mathcal{M}}
\def\cN{\mathcal{N}}
\def\cO{\mathcal{O}}
\def\cP{\mathcal{P}}
\def\cQ{\mathcal{Q}}
\def\cR{\mathcal{R}}
\def\cS{\mathcal{S}}
\def\cT{\mathcal{T}}
\def\cU{\mathcal{U}}
\def\cV{\mathcal{V}}
\def\cW{\mathcal{W}}
\def\cX{\mathcal{X}}
\def\cY{\mathcal{Y}}
\def\cZ{\mathcal{Z}}
\def\cd{\mathcal{d}}
\def\Mt{M_{t}}
\def\Mr{M_{r}}
\def\O{\Omega_{M_{t}}}
\newcommand{\figref}[1]{{Fig.}~\ref{#1}}
\newcommand{\tabref}[1]{{Table}~\ref{#1}}

\newcommand{\var}{\mathsf{Var}}
\newcommand{\fb}{\tx{fb}}
\newcommand{\nf}{\tx{nf}}
\newcommand{\BC}{\tx{(bc)}}
\newcommand{\MAC}{\tx{(mac)}}
\newcommand{\Pout}{p_{\mathsf{out}}}
\newcommand{\nnn}{\nn\\}
\newcommand{\FB}{\tx{FB}}
\newcommand{\TX}{\tx{TX}}
\newcommand{\RX}{\tx{RX}}
\renewcommand{\mod}{\tx{mod}}
\newcommand{\m}[1]{\mathbf{#1}}
\newcommand{\td}[1]{\tilde{#1}}
\newcommand{\sbf}[1]{\scriptsize{\textbf{#1}}}
\newcommand{\stxt}[1]{\scriptsize{\textrm{#1}}}
\newcommand{\suml}[2]{\sum\limits_{#1}^{#2}}
\newcommand{\sumlk}{\sum\limits_{k=0}^{K-1}}
\newcommand{\eqhsp}{\hspace{10 pt}}
\newcommand{\tx}[1]{\texttt{#1}}
\newcommand{\Hz}{\ \tx{Hz}}
\newcommand{\sinc}{\tx{sinc}}
\newcommand{\tr}{\mathsf{tr}}
\newcommand{\diag}{\mathrm{diag}}
\newcommand{\MAI}{\tx{MAI}}
\newcommand{\ISI}{\tx{ISI}}
\newcommand{\IBI}{\tx{IBI}}
\newcommand{\CN}{\tx{CN}}
\newcommand{\CP}{\tx{CP}}
\newcommand{\ZP}{\tx{ZP}}
\newcommand{\ZF}{\tx{ZF}}
\newcommand{\SP}{\tx{SP}}
\newcommand{\MMSE}{\tx{MMSE}}
\newcommand{\MINF}{\tx{MINF}}
\newcommand{\RC}{\tx{MP}}
\newcommand{\MBER}{\tx{MBER}}
\newcommand{\MSNR}{\tx{MSNR}}
\newcommand{\MCAP}{\tx{MCAP}}
\newcommand{\vol}{\tx{vol}}
\newcommand{\ah}{\hat{g}}
\newcommand{\tg}{\tilde{g}}
\newcommand{\teta}{\tilde{\eta}}
\newcommand{\heta}{\hat{\eta}}
\newcommand{\uh}{\m{\hat{s}}}
\newcommand{\eh}{\m{\hat{\eta}}}
\newcommand{\hv}{\m{h}}
\newcommand{\hh}{\m{\hat{h}}}
\newcommand{\Po}{P_{\mathrm{out}}}
\newcommand{\Poh}{\hat{P}_{\mathrm{out}}}
\newcommand{\Ph}{\hat{\gamma}}
\newcommand{\mat}[1]{\begin{matrix}#1\end{matrix}}
\newcommand{\ud}{^{\dagger}}
\newcommand{\C}{\mathcal{C}}
\newcommand{\nn}{\nonumber}
\newcommand{\nInf}{U\rightarrow \infty}

\title{Realizing In-Memory Baseband Processing for Ultra-Fast and Energy-Efficient 6G}
\author{
Qunsong Zeng, Jiawei Liu, Mingrui Jiang, Jun Lan, Yi Gong, Zhongrui Wang, Yida Li, Can Li, Jim Ignowski, and Kaibin Huang
\thanks{Q. Zeng, J. Liu, M. Jiang, Z. Wang, C. Li, and K. Huang are with The University of Hong Kong, Hong Kong SAR. J. Liu, J. Lan, Y. Gong, and Y. Li are with Southern University of Science and Technology, Shenzhen, China. J. Ignowski is with Hewlett Packard Enterprise, United States. Q. Zeng and J. Liu contributed equally to this work. Contact: K. Huang (huangkb@eee.hku.hk).}
}

\maketitle

\begin{abstract}
    To support emerging applications ranging from holographic communications to extended reality, next-generation mobile wireless communication systems require ultra-fast and energy-efficient baseband processors. Traditional complementary metal-oxide-semiconductor (CMOS)-based baseband processors face two challenges in transistor scaling and the von Neumann bottleneck. To address these challenges, in-memory computing-based baseband processors using resistive random-access memory (RRAM) present an attractive solution. In this paper, we propose and demonstrate RRAM-implemented in-memory baseband processing for the widely adopted multiple-input-multiple-output orthogonal frequency division multiplexing (MIMO-OFDM) air interface. Its key feature is to execute the key operations, including discrete Fourier transform (DFT) and MIMO detection using linear minimum mean square error (L-MMSE) and zero forcing (ZF), in one-step. In addition, RRAM-based channel estimation module is proposed and discussed. By prototyping and simulations, we demonstrate the feasibility of RRAM-based full-fledged communication system in hardware, and reveal it can outperform state-of-the-art baseband processors with a gain of $91.2\times$ in latency and $671\times$ in energy efficiency by large-scale simulations. Our results pave a potential pathway for RRAM-based in-memory computing to be implemented in the era of the sixth generation (6G) mobile communications.
\end{abstract}

\section{Introduction}
While the fifth generation (5G) mobile networks are being deployed, the sixth generation (6G) is under development all over the world to provide a new infrastructure for propelling the digital economy forward and realizing Society 5.0~\cite{ghosh2021ai}. The performance of 6G will be unprecedented as reflected in a set of target key performance indicators (KPIs), dictating a peak data rate to go beyond 100Gb/s, having a minimum latency 0.1ms, and achieving an energy efficiency of $10^{-12}$J/bit~\cite{dang2020should,akhtar2020shift,tataria20216g,vaezi2022cellular,rajatheva2020white}. This coined the term \emph{ultra-fast-and-energy-efficient} (UFEE) communication and will enable a wide range of emerging applications, for example, industrial automation~\cite{sodhro2020toward,popovski2019wireless}, tactile internet~\cite{holland2019ieee,promwongsa2020comprehensive,steinbach2018haptic}, holographic communications~\cite{clemm2020toward,wakunami2016projection}, and digital twin~\cite{nguyen2021digital,shen2021holistic}. Hence, this provides a strong motivation for 6G researchers to explore the largely unoccupied Terahertz (THz) spectrum~\cite{dang2020should,akhtar2020shift,tataria20216g,vaezi2022cellular,rajatheva2020white}. However, the required scaling up of baseband data rates to the hundreds of Gbps level will dramatically increase the power consumption and complexity of baseband processing, making it challenging to realize the 6G vision~\cite{sarieddeen2021overview,skrimponis2020power,rikkinen2020thz}. This is further exacerbated by the increasingly sophisticated communication techniques required, including large-scale \emph{multiple-input multiple-output} (MIMO), high-dimensional \emph{orthogonal frequency division multiplexing} (OFDM), and interference management. From 2G to 5G era, baseband processing demands have been satisfied largely by shrinking transistor size as governed by Moore’s law. Accordingly, the semiconductor industry has evolved from planar bulk \emph{Metal-Oxide-Semiconductor Field-Effect Transistors} (MOSFETs) to the recent 3D FinFETs and \emph{Gate All Around} (GAA) architectures to improve transistor performance and density in an integrated circuit (IC) chip~\cite{orji2018metrology}. However, this approach is facing increasing challenges as transistor size approaches the atomic limit~\cite{leiserson2020there}. In view of the Moore’s Law coming to an end, we propose the new paradigm called in-memory baseband processing for the post-Moore era, which adopts the emerging in-memory computing architecture instead of relying on transistor densification, to pave the way towards realizing the 6G UFEE connectivity.

Baseband processing and computing at large face two bottlenecks: the von Neumann bottleneck and the power wall, incurring large energy and footprint overheads. The former is due to data shuttling between the physically separated processing and storage units, resulting in significant latency and high energy consumption (e.g., 100-time more than digital logical circuits). In the latter, the increasing power density of transistors as the transistor size shrinks has created a “power wall” that limits practical processor frequency to $\sim$4 GHz since 2006~\cite{van2018method}, falling far short of the requirements for THz communications. 
In the past decade, researchers have started to improve computing latency and energy consumption by employing an architecture that co-locates data processing and storage, so called in-memory computing. 
Rather than making incremental improvements to conventional systems such as parallelism or memory bandwidth, in-memory computing takes a different approach by performing calculations where the data is located, thus fundamentally changing the von Neumann architecture~\cite{di2013parallel}. This method is similar to the way the human brain processes information in the networks of neurons and synapses, where there is no separation between computation and memory~\cite{indiveri2015memory}. In contrast to traditional computing schemes, in-memory computing eliminates latency and energy usage issues associated with the memory wall. However, this new architecture requires computational memory devices that can both store data and perform calculations simultaneously, usually by leveraging physical laws like Ohm's and Kirchhoff's laws in electrical circuits~\cite{ielmini2018memory}.
Emerging non-volatile memories such as \emph{resistive random-access memory} (RRAM) is touted as one of the most potential candidates for such computational memory devices~\cite{wang2020resistive}. It has been reported that parallel execution of a larger number (e.g., millions) of multiply-and-accumulate (MAC) operations for matrix vector multiplications (MVM) can be accomplished with extremely high energy-efficiency and low latency~\cite{sebastian2020memory}. 
This makes in-memory computing a UFEE solution for MVM intensive applications such as deep neural networks~\cite{hu2014memristor,ambrogio2018equivalent,burr2015experimental,yan2019rram,yao2020fully,zhang2021hybrid,wozniak2020deep,zhang2020artificial,duan2020spiking,li2020power,prezioso2015training,chen2019cmos,mahmoodi2019versatile,cai2019fully,li2018efficient,yao2017face} and linear algebra computation~\cite{sun2019solving,sun2020one,le2018mixed,zidan2018general}. Such advantages can naturally contribute to the trend of seamless integration of communication and artificial intelligence (AI) for the next-generation Internet-of-Things (IoTs). A new paradigm for communications called in-memory baseband processing, which adopts the emerging in-memory computing architecture and novel signal processing approach, are potential key factors to alleviate the challenges faced by researchers in realizing UFEE connectivity in the era of 6G.

\begin{figure*}[t!]
    \centering
    \includegraphics[width=0.85\textwidth]{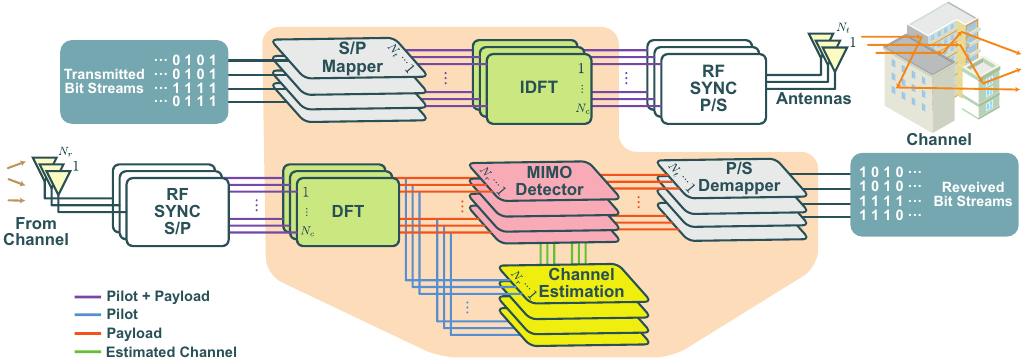}
    \caption{The architecture of RRAM-based transmitter: It consists of baseband processing modules [i.e., mapper and IDFT], RF modem, and an array of transmit antennas. Each layer represents a piece of RRAM-based circuit. The architecture of RRAM-based receiver: It is comprised of an array of receive antennas, RF front-end, and baseband processing modules [i.e., DFT, channel estimation, MIMO detection, and demapper].}
    \label{Fig.1}
\end{figure*}
6G will feature scaling up of different physical-layer technologies, for example, massive MIMO using large-scale antenna arrays~\cite{zhu2017hybrid} and OFDM comprising thousands of sub-carriers~\cite{tataria20216g}. The resultant baseband processing will involve frequent large-scale matrix operations. This motivates us to propose the new paradigm of in-memory baseband processing, which relocates the conventional digital operations to the analogue domain to achieve UFEE processing. In this paper, we present the design of an in-memory baseband processor for MIMO-OFDM which is a dominant air-interface technology for 5G-and-beyond~\cite{dang2020should,akhtar2020shift,tataria20216g,vaezi2022cellular,he2021cell}. The key novelty includes modules, namely OFDM demodulation, MIMO detection, and channel estimation, which are designed and implemented using in-memory computing approach based on Ta/TaO$_x$/Pt RRAM chip. The OFDM module implements the \emph{discrete Fourier transform} (DFT) using two RRAM crossbar arrays. Using such arrays to store DFT matrix enables one-step DFT operation, cutting down the power/latency overheads in conventional CMOS-based processor significantly. Furthermore, the required channel matrix inversion for MIMO detection is realized using a novel RRAM circuit featuring stability and easy mode switching, enabling the one-step operation. The performance of our design is evaluated using proof-of-concept prototypes for separate modules and a complete system by physical experiments, respectively, and simulations for a large-scale communication system. We show that the throughput and energy-efficiency can be boosted up to $91.2\times$ and $671\times$ respectively as compared to state-of-the-art CMOS-based baseband processors. 



\section{Overview of RRAM-based Baseband Processor}
In this paper, the baseband processor targets the MIMO-OFDM air interface, where a pair of multi-antenna transmitter and receiver communicate over a broadband channel. In broadband communications, frequency selective fading occurs when the channel having a coherence bandwidth is smaller than that of the signal causes its distortion. As a popular technology for coping with such fading as well as inter-symbol interference, OFDM is adopted to divide the whole bandwidth into $N_c$ orthogonal sub-channels. As a result, each sub-channel, say the $k$-th sub-channel, is a narrowband channel with $N_t$ transmit and $N_r$ receive antennas, modelled by a MIMO-channel matrix $\mathbf{H}^{(k)}\in\mathbb{C}^{N_r\times N_t}$ that is fixed within an OFDM symbol. The input-output relation of a MIMO system over the $k$-th sub-channel is given as 
\begin{equation}
    \mathbf{y}^{(k)}=\mathbf{H}^{(k)}\mathbf{x}^{(k)}+\mathbf{z}^{(k)},
\end{equation}
where $\mathbf{x}^{(k)}\in\mathbb{C}^{N_t\times1}$ consists of symbols at the $k$-th sub-carrier, $\mathbf{y}^{(k)}\in\mathbb{C}^{N_r\times1}$ comprises the received symbols at the $k$-th sub-carrier, and $\mathbf{z}^{(k)}$ represents the \emph{additive white Gaussian noise} (AWGN) in propagation.

The architectures of the RRAM-based transceiver are illustrated in Fig.~\ref{Fig.1}. Before baseband processing, the receiver still needs sampling at the RF front-end. Compared with a fast analogue-to-digital converter (ADC) design for traditional digital baseband processing, the proposed baseband processor features the direct processing of analogue-valued input signals so that the ADC can be replaced with a simpler sample-and-hold circuit. The baseband (information) processing starts at the mapper module in the transmitter that transforms bits into symbols and ends at the demapper module in the receiver that transforms the symbols back to bits. 
The digital modulation is chosen as 16 quadrature amplitude modulation (16-QAM) unless specified otherwise, which maps a 4-bit string to one of the 16 points on the constellation diagram. The bit stream is split into in-phase (denoted by $I$) and quadrature (denoted by $Q$) streams, associated with 0-degree and 90-degree phase shifts of the carrier wave, respectively. $I$ and $Q$ components are Gray encoded, i.e., neighbour points only differ in a single bit, to produce symbol points in the constellation. The system performance is evaluated by two metrics: i) The \emph{modulation error ratio} (MER) measures the dispersion of the constellation of the received symbols. To be specific, given total $M$ transmitted symbols, the definition of MER is 
\begin{equation}  
\text{MER}=10\log_{10}\left(\frac{\sum\limits_{m=1}^M\left(I_m^2+Q_m^2\right)}{\sum\limits_{m=1}^M\left[(I_m'-I_m)^2+(Q_m'-Q_m)^2\right]}\right)\text{dB},
\end{equation}
where $I_m$ and $Q_m$ denote the in-phase and quadrature components of the $m$-th transmitted symbol while $I_m'$ and $Q_m'$ denote the in-phase and quadrature components of the received symbol.
ii) The \emph{bit error ratio} (BER) is the number of bit errors divided by the total number of transmitted bits. To be specific, during the studied time interval, the BER is given by 
\begin{equation}
    \text{BER}=\frac{\rm \#~error~bits}{\rm \#~total~transmitted~bits}\times 100\%.
\end{equation}
In this work, we focus on the baseband processing between the mapper and demapper. The module in the transmitter performs \emph{inverse DFT} (IDFT). For the receiver, the three modules are DFT module, channel estimator, and MIMO detector. To reconcile signals and channels in the complex domain and the fact that RRAM devices store and compute real numbers, we propose to apply the mapping $\mathcal{R}:~\mathbb{C}^{K\times L}\to\mathbb{C}^{2K\times 2L}$ which transforms a complex matrix $\mathbf{A}\in\mathbb{C}^{K\times L}$ into a real matrix
    $\mathcal{R}(\mathbf{A})=\left[\begin{matrix}\Re(\mathbf{A})&-\Im(\mathbf{A})\\\Im(\mathbf{A})&\Re(\mathbf{A})\end{matrix}\right]\in\mathbb{R}^{2K\times 2L}$.
The complex vector is translated as the input voltages (or currents) for the RRAM array, with the mapping $\mathcal{T}:~\mathbb{C}^{K\times1}\to\mathbb{R}^{2K\times1}$ transforming a complex vector $\mathbf{x}\in\mathbb{C}^{K\times1}$ into a real vector
    $\mathcal{T}(\mathbf{x})=\left[\begin{matrix}\Re(\mathbf{x})\\\Im(\mathbf{x})\end{matrix}\right]\in\mathbb{R}^{2K\times1}$.
Such transformations enables the equivalent computation involving complex matrices and vectors.

\begin{figure}[t!]
    \centering
    \includegraphics[width=0.95\columnwidth]{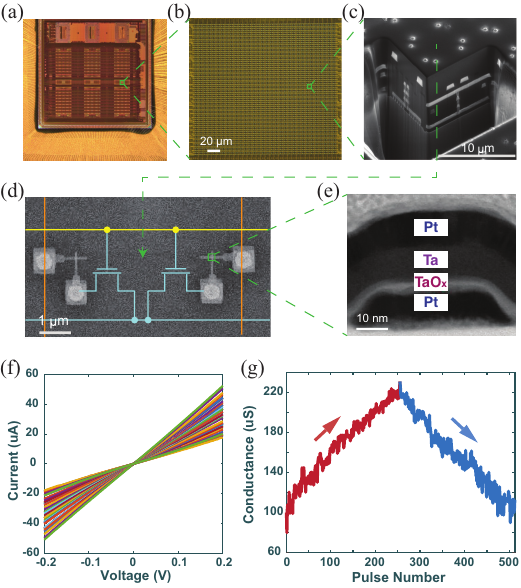}
    \caption{Integrated RRAM chip and measurements. (a) A photo of a wire-bounded integrated RRAM chip, which contains three 64 $\times$ 64 1T1R crossbar arrays, row MUXes, column MUXes, transimpedance amplifiers (TIA), sample-and-hold, and ADC. (b) Optical image of a 64 $\times$ 64 RRAM crossbar array. (c) The cross-section view of the integrated chip with CMOS circuits at the bottom, inter-connection in the middle, and metal through-hole on the surface used for RRAM and back-end process integration. (d) Top view of four cross-point RRAM devices. (e) The TEM image of the RRAM device. (f) Ohmic behaviour of RRAM devices. The linear I-V relationship is illustrated at different conductance states under different read voltages (-0.2$\sim$0.2V). (g) The conductance modulation characteristic of the RRAM device. A train of voltage pulses (pulse width 10ns) are applied for the RRAM conductance modulation measurements. The magnitude of voltages starts at 0.60V and grows to 0.70V smoothly for potentiation, while it starts at -0.50V and drops to -0.65V gradually for depression. The cycle-to-cycle variations are $4.41\%$ during potentiation and $5.44\%$ during depression, respectively. The conductance ranges from 79.93$\mu$S to 230.99$\mu$S in the behavioral measurement.}
    \label{Fig.2}
\end{figure}

We use the Ta/TaO$_x$/Pt-based RRAM arrays as the hardware accelerators for its compatibility with traditional CMOS process and reliable electrical characteristics. Details of the RRAM array fabrication and integration are described in Appendix~\ref{Appendix: RRAM device fabrication}. The wire-bonded integrated RRAM chip that we used to implement the baseband processor modules is shown in Fig.~\ref{Fig.2}(a), which contains three $64 \times 64$ RRAM crossbar arrays and one of them is shown in Fig.~\ref{Fig.2}(b). The 50nm $\times$ 50nm Ta/TaO$_x$/Pt RRAMs are integrated with back-end-of-the-line (BEOL) processing on top of the control peripheral circuits (see Fig.~\ref{Fig.2}(c)). The peripheral control circuits are implemented with a commercial 180nm technology integrated chip, among which the access transistors are highlighted in Fig.~\ref{Fig.2}(d). Such one-transistor-one-resistor (1T1R) array architecture avoids the sneak current issue during RRAMs’ conductance programming and allows each cell in the array to be accessed independently \cite{yu2016emerging}. The cross-sectional transmission electron microscopy (TEM) image of the RRAM device is shown in Fig.~\ref{Fig.2}(e). As a non-volatile analogue device, our RRAM device exhibits linear Ohmic behaviour (see Fig.~\ref{Fig.2}(f)) to ensure accurate in-memory computing. For matrix mapping, the conductance programming of the fabricated RRAM device can be achieved by applying a train of positive pulses (0.60$\sim$0.70V/10ns) for potentiation, and continuous negative pulses (-0.50$\sim$-0.65V/10ns) for depression (see Fig.~\ref{Fig.2}(g)). 


\begin{figure}[t!]
    \centering
    \subfigure[DFT module design]{
    \includegraphics[width=0.98\columnwidth]{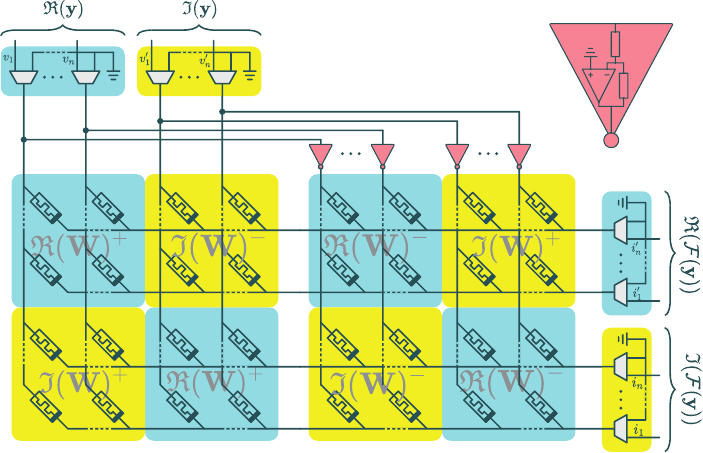}}
    \subfigure[Real mapped DFT matrix ($\mu$S)]{
    \includegraphics[width=0.45\columnwidth]{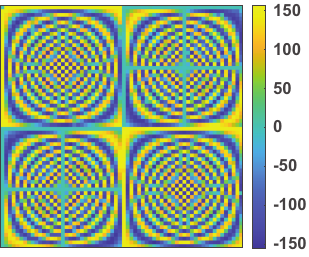}}
    \subfigure[Error matrix]{
    \includegraphics[width=0.45\columnwidth]{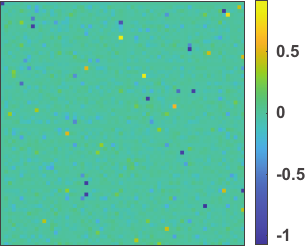}}
    \caption{Orthogonal frequency-division multiplexing modules. (a) The architecture of RRAM-based DFT module. The DFT matrix $\mathbf{W}$ is stored in the RRAM array and the input signal $\mathbf{y}$ is translated as the voltages to be applied to the array. The elements (and signals) in real and imaginary domains are highlighted by different colours. The module performs the DFT over signal $\mathbf{y}$ and the read drivers get the result $\mathcal{F}(\mathbf{y})=\mathbf{W}\mathbf{y}$. The design of inverting amplifier is presented at the upper right corner. (b, c) In the experiment, the real mapping of the DFT matrix is scaled and programmed into two 64 $\times$ 64 RRAM arrays in our integrated RRAM chip: (b) conductance matrix and (c) corresponding error matrix, each element of which refers to the ratio (experimental conductance – target conductance)/target conductance [note: the value is not in percentage form].}
    \label{Fig.3}
\end{figure}


\section{Orthogonal Frequency-Division Multiplexing Module}
The RRAM-based DFT module is illustrated in Fig.~\ref{Fig.3}(a), where data are modulated onto non-interfering sub-carriers in the frequency domain. The transformation between the time and frequency domains involves IDFT/DFT operations. For the received block of symbols $\mathbf{y}$, the DFT of which can be represented as an $N_c$-length vector: $\mathcal{F}(\mathbf{y})=\mathbf{W}\mathbf{y}$, where $\mathcal{F}(\cdot)$ denotes the DFT operation and $\mathbf{W}$ is the DFT matrix. In the circuit design, the real mapping of DFT matrix, $\mathcal{R}(\mathbf{W})$, is scaled into the RRAM devices’ conductance range and stored as the difference between two arrays. The received signal $\mathbf{y}$ is translated to the input voltages $\mathcal{T}(\mathbf{y})$ for the array. The module computes the DFT of $\mathbf{y}$, and the current outputs are the scaled real vector mapping $\mathcal{T}(\mathcal{F}(\mathbf{y}))$. The detailed discussion on the hardware implementation of this module is provided in Appendix~\ref{Appendix: DFT/IDFT and channel estimator circuit}. Compared with conventional approaches based on \emph{fast Fourier transform} (FFT) algorithms~\cite{cochran1967fast}, the RRAM-based design features the dramatic reduction of computational complexity of from $O(N_c\log N_c)$ for FFT to just a one-step (i.e., $O(1)$) operation. This makes it possible to overcome the bottleneck of high complexity of DFT in baseband processing for the next-generation large-scale OFDM communications.

\begin{figure}[t!]
    \centering\hspace{-3mm}
    \subfigure[Digital processor]{
    \includegraphics[width=0.33\columnwidth]{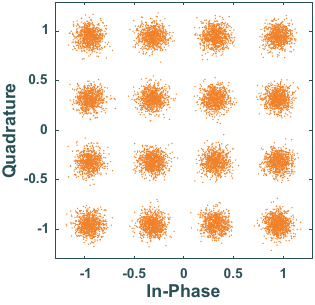}}\hspace{-2mm}
    \subfigure[RRAM processor]{
    \includegraphics[width=0.33\columnwidth]{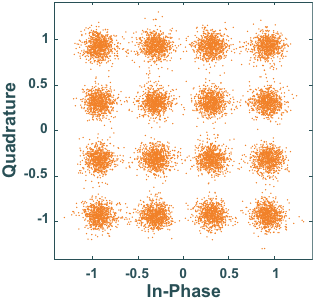}}\hspace{-2mm}
    \subfigure[Defection corrected]{
    \includegraphics[width=0.33\columnwidth]{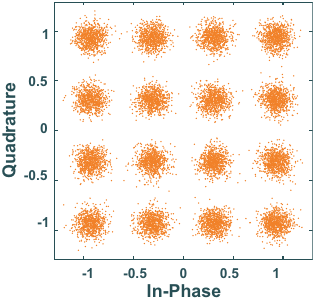}}
    \subfigure[BER performance]{
    \includegraphics[width=0.45\columnwidth]{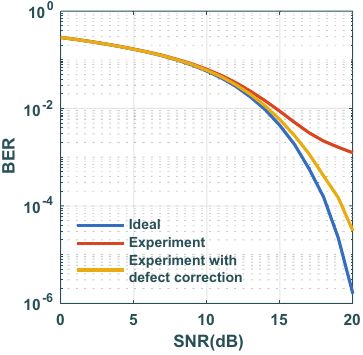}}
    \subfigure[MER performance]{
    \includegraphics[width=0.45\columnwidth]{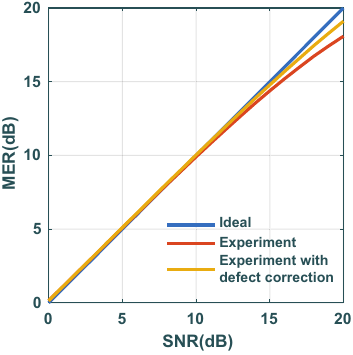}}
    \caption{Experimental performance of the DFT module in the demonstrated OFDM system. At the receiver, the constellation diagram is recovered using (a) digital processor, (b) the RRAM processor experimental results, (c) RRAM processor results compensated by defective correction. (d, e) Communication performance of digital processor, RRAM-implemented DFT and defection corrected RRAM-DFT under different channel conditions.}
    \label{Fig.4}
\end{figure}

In this section, a single-antenna OFDM system with 32 sub-carriers is demonstrated. The conductance mapping of the DFT matrix to RRAM array is scaled to fit the RRAM devices’ conductance range, which are programmed into two arrays. The subtraction of the conductance matrices of these two arrays, in the form of differential pairs, and the corresponding error matrix are presented in Fig.~\ref{Fig.3}(b) and Fig.~\ref{Fig.3}(c). The complete signal processing path in the prototypical RRAM-based OFDM system is described as follows. For the transmitter, a message in bits is firstly modulated into 16-QAM symbols, and then transformed from the frequency domain into the time domain by IDFT. After adding a cyclic prefix, the OFDM symbols are transmitted over the channel towards the receiver. At the receiver, after removing the cyclic prefix, the RRAM-based DFT is performed to transform the received signal back to the frequency domain, where the symbols are then demodulated into bits to recover the message. The performance of the receiver with RRAM-based DFT module is experimentally characterized over the wireless channel of receive signal-noise-ratio (SNR) being 20dB. As a benchmark, the constellation diagram recovered by the digital processor using a double-precision floating-point DFT matrix is shown in Fig.~\ref{Fig.4}(a). In this case, the distortion of demodulated symbols is measured by MER 20dB at which no bit errors occur. For our RRAM-implemented DFT prototype, the measured constellation diagram is shown in Fig.~\ref{Fig.4}(b) with MER dropping to 18dB while the BER growing to 0.00146. The communication performance is affected by both the channel noise and RRAM devices’ imperfections. Compared with the results from a digital processor, the performance loss of our experimental RRAM-implemented DFT module comes from the imperfections of RRAM devices in the array, including defections and programming errors. To reveal the effect of the defective devices, we compensate the defections by post-processing. To be specific, we define the defection matrix $\mathbf{W}_{\rm defection}$ as the compensatory conductance matrix of stuck-on and stuck-off RRAM devices. We then perform the multiplication operation $\mathbf{W}_{\rm defection} \mathbf{y}$ in computer and add the result into the experimental outcome to obtain the defection-corrected result: $\mathbf{x}=\mathbf{W}_{\rm RRAM} \mathbf{y}+\mathbf{W}_{\rm defection} \mathbf{y}$. Leveraging this method, we rectify the experimental constellation diagram from RRAM-implemented DFT module as shown in Fig.~\ref{Fig.4}(c), whose BER is ameliorated by an order of magnitude. Moreover, we explore the performance of our design with different transmission powers (i.e., different SNRs) as shown in Fig.~\ref{Fig.4}(d) and Fig.~\ref{Fig.4}(e). We observe that the performance differences between digital processor and our experimental RRAM-implemented DFT module are insignificant for a noisy channel. However, the differences can be noticeable for cleaner channels where the imperfections of RRAM devices in the array deteriorate the communication performance. The defective devices play a destructive role in the baseband processing and tackling this issue brings benefit to the enhancement of performance.

\begin{figure}[t!]
    \centering
    \subfigure[MIMO module design]{
    \includegraphics[width=0.99\columnwidth]{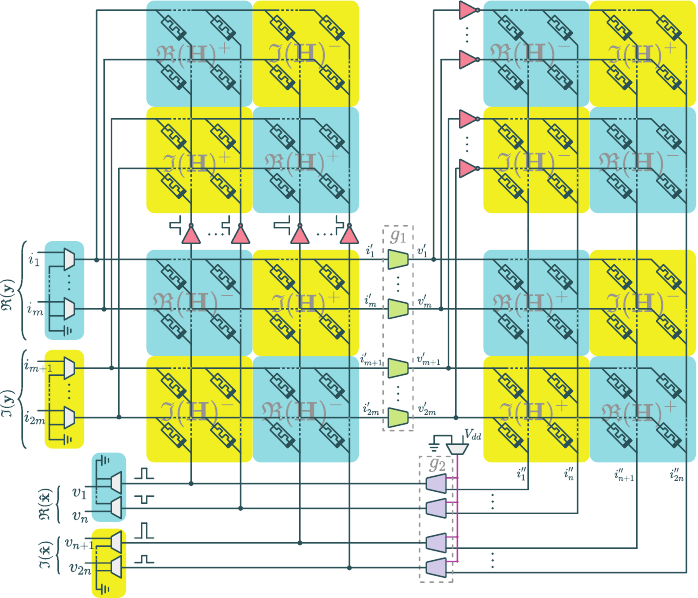}}
    \subfigure[TIAs]{
    \includegraphics[height=3cm]{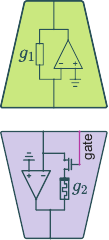}}
    \subfigure[Real mapped channel ($\mu$S)]{
    \includegraphics[height=3cm]{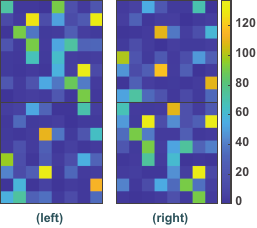}}
    \subfigure[Error matrix]{
    \includegraphics[height=3cm]{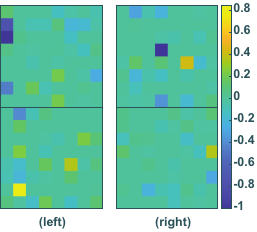}}
    \caption{Multiple-input and multiple-output modules. (a) The architecture of RRAM-based MIMO detection module. The channel matrix $\mathbf{H}$ is stored in the four RRAM arrays as marked in the figure and the input signal $\mathbf{y}$ is scaled and translated as the input currents. The elements (and signals) in real and imaginary domains are highlighted by different colours. (b) The transistor controls whether L-MMSE or ZF modules is adopted. When the gate is grounded, the circuit performs ZF detection. Otherwise, L-MMSE is selected. In addition, to adapt to environments with different SNRs, the feedback conductance of the operational amplifiers is tuneable. (c, d) In the experiment, the real mapped channel matrix is scaled and programmed into RRAM arrays in our integrated RRAM chip: (c) conductance matrix and (d) corresponding error matrix, each element of which refers to the ratio (experimental conductance – target conductance)/target conductance.}
    \label{Fig.5}
\end{figure}

\begin{figure}[t!]
    \centering\hspace{-3mm}
    \subfigure[Digital processor]{
    \includegraphics[width=0.33\columnwidth]{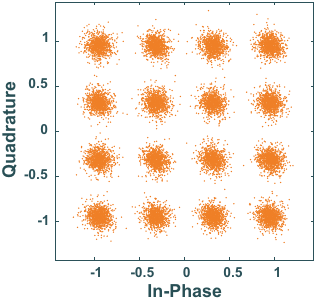}}\hspace{-2mm}
    \subfigure[RRAM processor]{
    \includegraphics[width=0.33\columnwidth]{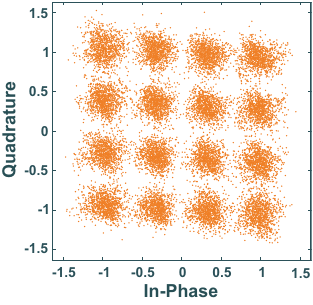}}\hspace{-2mm}
    \subfigure[Defection corrected]{
    \includegraphics[width=0.33\columnwidth]{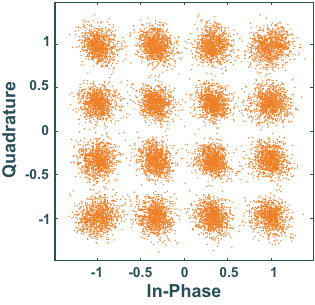}}
    \subfigure[BER performance]{
    \includegraphics[width=0.45\columnwidth]{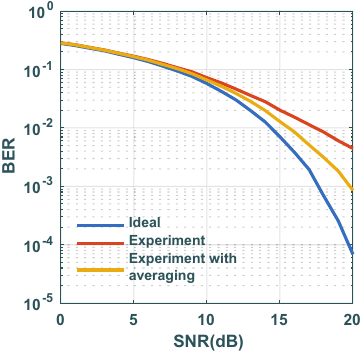}}
    \subfigure[MER performance]{
    \includegraphics[width=0.45\columnwidth]{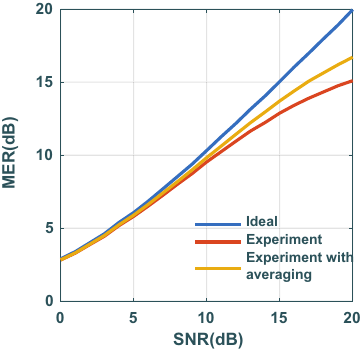}}
    \caption{Experimental performance of the MIMO detection module in the demonstrated MIMO system. At the receiver, the constellation diagram is recovered using (a) digital processor, (b) RRAM-implemented MIMO detection module, (c) RRAM MIMO detection module averaged from two implementations. (d, e) Communication performance of digital processor, RRAM-implemented MIMO detection, and averaged RRAM-pair implementation under different channel conditions.}
    \label{Fig.6}
\end{figure}

\section{Multiple-Input and Multiple-Output Detection Module}
The RRAM-based MIMO detection module is illustrated in Fig.~\ref{Fig.5}(a) and Fig.~\ref{Fig.5}(b), which spatially multiplexes multiple parallel data-streams. This scales up the system throughput since different symbols are simultaneously transmitted over different antennas. Exploiting the unique channel between each pair of transmit and receive antennas allows each transmitted symbol to be recovered through the module of MIMO detection. In practice, two linear detectors are widely used, namely \emph{linear minimum mean square error} (L-MMSE) and \emph{zero forcing} (ZF) detectors. They reverse the signal distortion by propagation through a MIMO channel by channel equalization. To be specific, given the channel matrix $\mathbf{H}$, the L-MMSE detector minimizes the mean squared error in the estimate of $\mathbf{x}$ among all linear detectors. The recovered signal vector is given by $\hat{\mathbf{x}}=\left(\mathbf{H}^{\sf H}\mathbf{H}+\frac{1}{\sf SNR}\mathbf{I}\right)^{-1}\mathbf{H}^{\sf H}\mathbf{y}$, where $\mathbf{y}$ is the received signal vector at the receiver. In hardware implementation, the equivalent real-value channel matrix $\mathcal{R}(\mathbf{H})$ is scaled and written into the RRAM arrays in the way as illustrated in Fig.~\ref{Fig.5}(a), and the received signal $\mathbf{y}$ is scaled and translated to the input voltages $\mathcal{T}(\mathbf{y})$. The output voltages are the real vector mapping $\mathcal{T}(\hat{\mathbf{x}})$, and the detailed analysis of this circuit is provided in Appendix~\ref{Appendix: L-MMSE/ZF MIMO Detector Circuit}. To cope with heterogeneous propagation environments with different SNRs, the feedback conductance of operational amplifiers can be represented using a RRAM device as shown in Fig.~\ref{Fig.5}(b). Our design also applies to ZF detection (see Appendix~\ref{Appendix: L-MMSE/ZF MIMO Detector Circuit}) which solves the least square problem and gives the recovered signal vector as $\hat{\mathbf{x}}=\left(\mathbf{H}^{\sf H}\mathbf{H}\right)^{-1}\mathbf{H}^{\sf H}\mathbf{y}$. As shown in Fig.~\ref{Fig.5}(b), the transistor dictates whether L-MMSE or ZF is applied. If the channel matrix is square, i.e., $N_t=N_r=N$, the computational complexity of conventional matrix inversion is $O(N^3)$. The complexity increases rapidly as the number of transmit/receive antennas grows. On the contrary, the proposed MIMO detection performs the computation in just a single step (i.e., $O(1)$), presenting a promising solution for efficient detection in the 6G massive MIMO communication.

We experimentally demonstrated the RRAM-based narrowband MIMO system with 4 transmit antennas and 4 receive antennas. The real mapped channel matrix $\mathcal{R}(\mathbf{H})$ is scaled and programmed into the RRAM arrays (see Fig.~\ref{Fig.5}(c)). The programming error is presented in Fig.~\ref{Fig.5}(d) . The experimental results of the constellation diagrams from L-MMSE detection for a noisy channel of SNR being 20dB. To begin with, the constellation diagram recovered by the digital processor is shown in Fig.~\ref{Fig.6}(a) as a benchmark. For our RRAM-implemented MIMO detection module, the measured MER drops 4dB compared to the digital counterpart, inducing more bit errors (see Fig.~\ref{Fig.6}(b)). The performance loss comes from the programming noise of RRAM devices whose effects on the circuit are twofold: i) the imprecision of the channel matrix representation; ii) the imbalance of the left and right channel matrices. To reduce the effect of the programming noise in RRAM devices, we use two RRAM devices to represent one real value which supresses the variance of random noise. By this means, we find the dispersion of constellation points (see Fig.~\ref{Fig.6}(c)) becomes less severe and the BER reduces by an order of magnitude. Furthermore, as shown in Fig.~\ref{Fig.6}(d) and Fig.~\ref{Fig.6}(e), we study the performance differences between digital processor and our RRAM-implemented MIMO detection module under different channel conditions. We observe subtle differences between them in the low SNR regime while the divergence becomes noticeable in the high SNR regime. The performance loss of our RRAM-implemented MIMO detection module can be relieved by representing channel matrix using more RRAM devices to reduce the programming noise. 

\begin{figure*}[t!]
    \centering
    \includegraphics[width=0.98\textwidth]{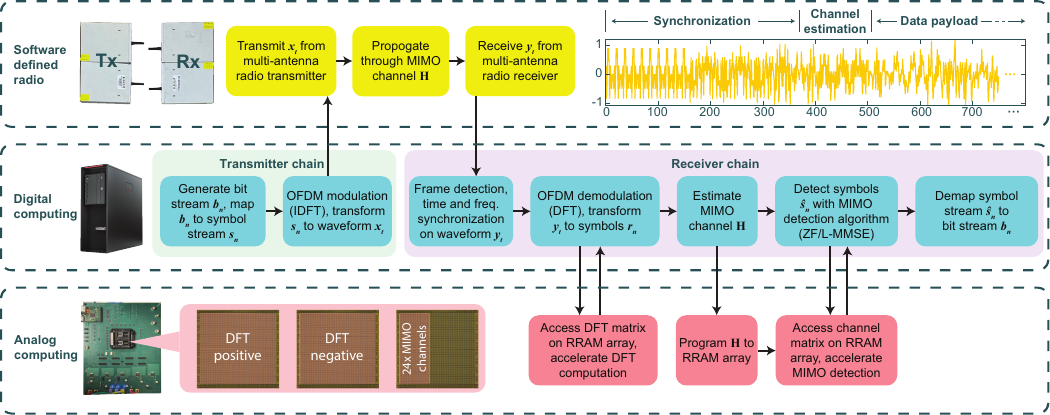}
    \caption{Proof-of-concept in-memory baseband processing experiment. A small-scale MIMO-OFDM system (32 sub-carriers, 2 transmit and 2 receive antennas) is validated in the experiment using SDRs, workstation and RRAM crossbar arrays. A schematic of the experimental system is presented. The realistic communication system is realized by SDR platform. Our in-memory computing test board provides physical measurement for the computation of DFT and MIMO detection modules. The computer bridges different platforms and controls the dataflow.}
    \label{Fig.7}
\end{figure*}

\begin{figure}[t!]
    \centering\hspace{-3mm}
    \subfigure[Transmitted Symbols]{
    \includegraphics[width=0.33\columnwidth]{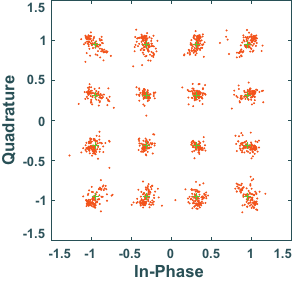}}\hspace{-2mm}
    \subfigure[Digital Processor]{
    \includegraphics[width=0.33\columnwidth]{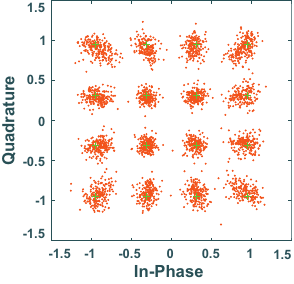}}\hspace{-2mm}
    \subfigure[RRAM Processor]{
    \includegraphics[width=0.33\columnwidth]{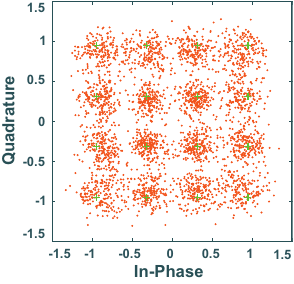}}
    \caption{Proof-of-concept in-memory baseband processing experimental results: (a) Constellation diagram of transmitted symbols. (b, c) Constellation diagrams of the symbols recovered at the receiver from (b) digital baseband processing and (c) RRAM-implemented in-memory baseband processing.}
    \label{Fig.8}
\end{figure}

\section{Channel Estimation Module}
To acquire the \emph{channel state information} (CSI) needed for MIMO detection, the channel matrix is estimated at the receiver using pilot signals that are sent by the transmitter and known a priori to the receiver. Many data symbols can be transmitted between two pilot signals separated by channel coherent time, amortizing the overhead of channel training. A larger ratio between data and pilot symbols improves the system throughput at the cost of adaptive to time-varying channels. The transmitted training matrix $\mathbf{P}\in\mathbb{C}^{N_t\times N_t}$ is known by the receiver, while the actual received matrix is $\mathbf{S}\in\mathbb{C}^{N_r\times N_t}$. By choosing the pilot signal as a unitary matrix~\cite{hampton2013introduction}, i.e., $\mathbf{P}\mathbf{P}^{\sf H}=\mathbf{I}$, the channel matrix estimated by \emph{maximum likelihood} (ML) or \emph{least square} (LS) is given as $\hat{\mathbf{H}}=\mathbf{S}\mathbf{P}^{\sf H}$. In the RRAM-based channel estimation module, the real mapped training matrix $\mathcal{R}(\mathbf{P})$ is stored in the RRAM array. Each row vector of the real mapped received matrix $\mathcal{R}(\mathbf{S})$ is translated to the supplied input voltages. The computation can be completed by $2N_r$ read pulses while the complexity is $O(N_r N_t^2)$ for traditional processors.

When ready, the row vectors of the estimated channel matrix are sequentially written into the RRAM array implementing the MIMO detector. We evaluate the performance of different writing process in terms of system latency. To this end, a mathematical model is developed to facilitate latency analysis for programming a 1T1R array as elaborated in Appendix~\ref{Appendix: latency analysis}. Consider the writing process using a train of pulses to program an $N\times N$ array in the row-by-row manner. It can be proved that the expected writing latencies of write-without-verification and write-with-verification schemes scale with the array size in the way no faster than $O(N\sqrt{\ln N})$ and $O(N\ln N)$, respectively. This contributes to the most latency in our design. For comparison, the computational complexity is $O(N^3)$ for traditional digital processors.


\section{Performance Evaluation of the Complete System}
Recall that we consider the MIMO-OFDM air interface where a transmitter/receiver integrates the RRAM-based OFDM and MIMO modules. The modules are separately validated in previous sections. Here, we report a system-level demonstration of a MIMO-OFDM system with 32 sub-carriers and 2 $\times$ 2 antennas for proof-of-concept. The RF chains are physically implemented using software defined radio (SDR) platform, which provides the realistic MIMO communication links over-the-air. The digital logic on a workstation regulates data generation, executes frame synchronization algorithms, orchestrates the operation of the other two platforms, and controls the system data flow. The system schematic is presented in Fig.~\ref{Fig.7}. The workflow is described as follows. The workstation randomly generates bit stream and maps bits to symbols in 16-QAM constellation diagram. The symbols are then transformed to time domain waveforms by OFDM modulation, which are fed in the SDR transmitter to radiate the wireless signals at carrier frequency into the air by the two transmit antennas. The constellation diagram of symbols in the data payload when they are emitted from the transmitter is presented in Fig.~\ref{Fig.8}(a). The dispersion of constellation points results from the thermal noise in transmitter circuit and the non-ideality of the RF components (e.g., nonlinear power amplifier response). The RF signal is captured by the two receive antennas at the SDR receiver, and down-converted to baseband signal with the locally generated carrier frequency. After that, the channel matrix is estimated using pilot symbols and programmed into RRAM arrays as mentioned. The data payloads are then processed using RRAM-implemented modules (i.e., DFT and MIMO detection), and the recovered constellation diagram is presented in Fig.~\ref{Fig.8}(c) with MER being 12.83dB. For comparison, the constellation points from a digital baseband processor are given in Fig.~\ref{Fig.8}(b) whose MER is 17.43dB. The performance loss of our RRAM-implemented system mainly comes from the defective RRAM devices and programming errors, which are compensatory as discussed. Our demonstration proves the feasibility of system-level RRAM-based in-memory baseband processing in a real wireless communication system.

\begin{figure}[t!]
    \centering
    \includegraphics[width=\columnwidth]{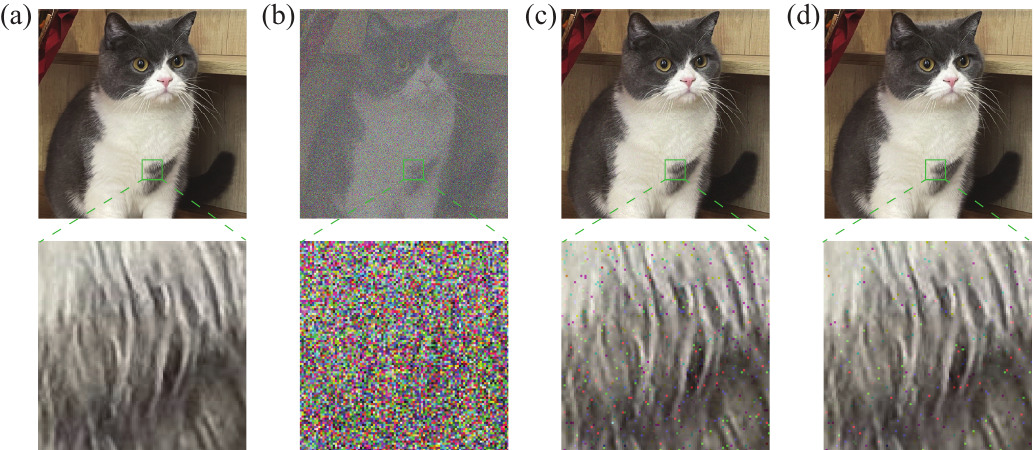}
    \caption{An illustration of the communication performance of transmitting an image over a noisy channel. (a) The original image. (b, c) The recovered images are from RRAM-based baseband processor where the RRAM arrays at the MIMO detector are programmed by (b) write-without-verification and (c) write-with-verification schemes. (d) Benchmark: software result.}
    \label{Fig.9}
\end{figure}

In the following, we perform the simulations of a large-scale RRAM-accelerated communication system corresponding to the standard of 5G new radio (NR) (see Table~\ref{Table: simulation parameters}). The simulation of RRAM array programming is based on the RRAM model calibrated using the experimentally acquired device properties such as the Ohmic behaviour (see Fig.~\ref{Fig.2}(f)) and the evolution of the conductance with voltage pulses (see Fig.~\ref{Fig.2}(g)). Both the cycle-to-cycle variations and read noise during RRAM programming are included in our simulations. 
Since the transmitter is much simpler than the receiver, we focus on the RRAM-based receiver in the remainder of this section. To visually demonstrate the performance of our designed in-memory baseband processor, we consider the specific task of uncoded transmission of an image as shown in Fig.~\ref{Fig.9}(a). The image recovered at the receiver are presented in Fig.~\ref{Fig.9}(b) and Fig.~\ref{Fig.9}(c) where RRAM devices are programmed using writing without and with verification schemes, respectively. As a benchmark for comparison, the image resulting from a digital baseband processor is shown in Fig.~\ref{Fig.9}(d). One can observe that the performance of the write-without-verification scheme is poor while the other scheme with verification performs similarly as the ideal processor. To quantify the performance, we present the relation between MER (and BER) and SNR for both schemes as shown in Fig.~\ref{Fig.10}(a) (and Fig.~\ref{Fig.10}(b)). The simulation results are aligned with the earlier observation and show that write-with-verification scheme outperforms the other in terms of communication performance. From the perspective of latency and energy efficiency, the performance is compared in Fig.~\ref{Fig.10}(c) and Fig.~\ref{Fig.10}(d). As shown in Table~\ref{Table: comparison with CMOS-based digital processor}, the throughput and energy-efficiency of the proposed RRAM-based in-memory baseband processing exceed those of any reported CMOS-based digital processors~\cite{snapdragonx65,fryza2014power,chen2022507,liu2018high,tang20192}. 
For example, under the specifications in Table~\ref{Table: simulation parameters}, the throughput and energy efficiency can achieve up to 160.8Gb/s and 4637Gb/J, exceeding state-of-the-art digital counterparts in the literature by more than 329$\times$ and 671$\times$, respectively. Moreover, by reasonable conversion, our design is estimated to supports a throughput 91.2$\times$ higher than one of the state-of-the-art commercial modems, i.e., Qualcomm Snapdragon X65. Underpinning the improvements is the ultra-fast one-step baseband processing after channel estimation such that the baseband latency mostly comes from programming the RRAM arrays of MIMO detection module at the beginning of the frame. In contrast, for CMOS-based digital processors, data symbols are processed by executing the DFT (or FFT) and MIMO-detection algorithms using digital logic circuits, both suffering from high complexity as discussed. Next, there exists a tradeoff between communication performance and latency, i.e., higher performance requires better programming accuracy and thus longer latency. On the one hand, the write-without-verification scheme shows lower latency but poor communication performance in terms of BER and MER, a result of the intrinsic stochasticity of RRAM. On the other hand, RRAM with more states can achieve higher precision but possibly more pulses are needed to reach the target conductance value. 

\begin{table*}[t!]
\centering
\caption{Parameters for Large-Scale MIMO-OFDM System Simulation}
\begin{tabular}{|c|c|c|c|c|c|}
\hline
Parameter & \# sub-carriers & \# Tx antennas & \# Rx antennas & \# Pilot symbols & \# Symbols/frame \\ \hline
Notation  & $N_c$           & $N_t$          & $N_r$          & $N_t$            & $M$                  \\ \hline
Value     & 1024            & 4              & 4              & 4                & 14$\times$160        \\ \hline
\end{tabular}
\label{Table: simulation parameters}
\end{table*}

\begin{table*}[t!]
\centering
\caption{Comparison with CMOS-based Digital Processors}
\begin{tabular}{|c|c|c|c|c|c|}
\hline
Processor                                                                                                                     & Technology & \begin{tabular}[c]{@{}c@{}}Latency\\ (ms)\end{tabular}    & \begin{tabular}[c]{@{}c@{}}Energy\\ (mJ)\end{tabular}    & \begin{tabular}[c]{@{}c@{}}Communication\\ Throughput (Gb/s)\end{tabular} & \begin{tabular}[c]{@{}c@{}}Energy\\ Efficiency (Gb/J)\end{tabular} \\ \hline
\begin{tabular}[c]{@{}c@{}}Qualcomm\\ Snapdragon X65 \cite{snapdragonx65}\end{tabular}                                                   & 4 nm       & \textless{}10   & N/A             & \textless{}5                                                              & N/A                                                                \\ \hline
\begin{tabular}[c]{@{}c@{}}TMS320C6678\\ 8-core digital signal processor \cite{fryza2014power}\end{tabular}                               & 20 nm      & 589.9           & 6548            & 0.0621                                                                    & 0.0056                                                             \\ \hline
\begin{tabular}[c]{@{}c@{}}Domain adaptive processor\\ 16$\times$DAP in literature \cite{chen2022507}\end{tabular}                     & 12 nm      & 74.95           & 6547            & 0.4888                                                                    & 0.0056                                                             \\ \hline
\begin{tabular}[c]{@{}c@{}}Combined digital baseband modules: \\ FFT in \cite{liu2018high} and MIMO detector in \cite{tang20192}\end{tabular} & 65 nm      & 50.17           & 5.3024          & 0.7303                                                                    & 6.9091                                                             \\ \hline
\textbf{\begin{tabular}[c]{@{}c@{}}Proposed RRAM-based \\ Baseband Processor\end{tabular}}                                    & \textbf{-} & \textbf{0.2278} & \textbf{0.0079} & \textbf{160.8}                                                            & \textbf{4637}                                                      \\ \hline
\end{tabular}
\label{Table: comparison with CMOS-based digital processor}
\end{table*}

\begin{figure}[t!]
    \centering
    \subfigure[MER v.s. SNR]{
    \includegraphics[width=0.45\columnwidth]{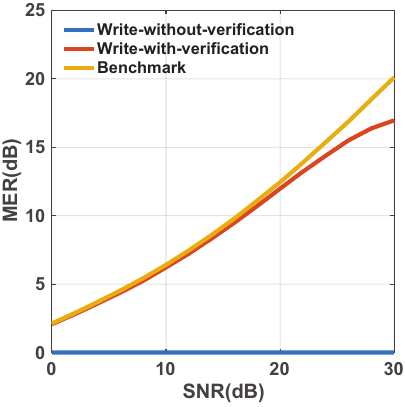}}
    \hspace{2mm}
    \subfigure[BER v.s. SNR]{
    \includegraphics[width=0.45\columnwidth]{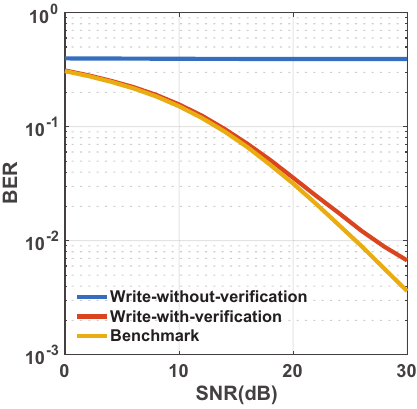}}
    \subfigure[Latency v.s. \#Antennas]{
    \includegraphics[width=0.45\columnwidth]{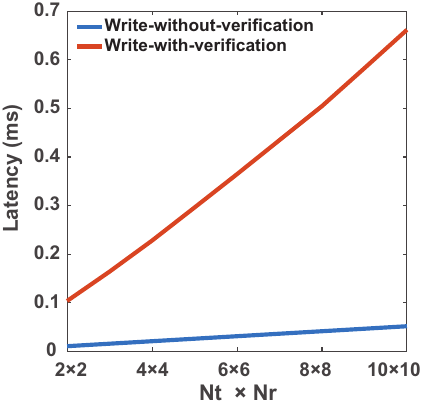}}
    \hspace{2mm}
    \subfigure[Energy v.s. \#Antennas]{
    \includegraphics[width=0.45\columnwidth]{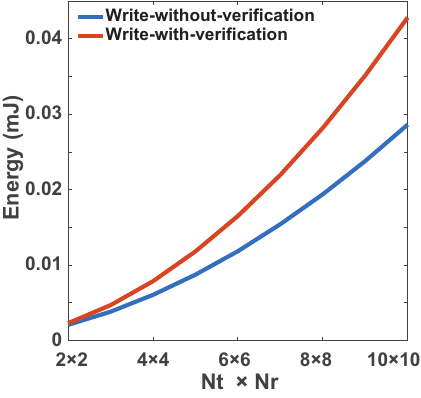}}
    \caption{Performance evaluation of in memory MIMO-OFDM baseband processing using experimental RRAMs. The simulations target a large-scale MIMO-OFDM system of 1024 sub-carriers, 4 transmit antennas and 4 receive antennas unless specified otherwise. The behavioural model of RRAM devices comes from the experimental testing of our fabricated RRAM devices. The simulation curves are averaged over 200 trials to eliminate the randomness of channel and RRAM devices. (a, b) Under different channel conditions, the resultant (a) MER and (b) BER from the three schemes. (c, d) For writing with and without verification schemes, the (c) latency and (d) energy are evaluated in terms of different MIMO sizes.}
    \label{Fig.10}
\end{figure}


\begin{table*}[t!]
\centering
\caption{Parameters of Memristors' Behavioral Models for the Simulations}
\begin{tabular}{|c|c|c|c|c|c|c|}
\hline
Memristor Device   & Mechanism   & \begin{tabular}[c]{@{}c@{}}Pulse\\ width\end{tabular} & \begin{tabular}[c]{@{}c@{}}State\\ number\end{tabular} & \begin{tabular}[c]{@{}c@{}}Cycle-to-cycle\\ vatiation\end{tabular} & $G_{\max}/G_{\min}$ & \begin{tabular}[c]{@{}c@{}}Operation\\ voltages\end{tabular} \\ \hline
\begin{tabular}[c]{@{}c@{}}Ta/TaO$_x$/Pt\\ our RRAM\end{tabular}
& Filament & 10 ns & 256 & \begin{tabular}[c]{@{}c@{}}4.41\% (P)\\ 5.44\% (D)\end{tabular}  
& 230.99/79.93 $\mu$S  & 0.65/-0.575 V \\ \hline
\begin{tabular}[c]{@{}c@{}}TiN/HZO/SiO$_2$/Si\\ FeFET \cite{jerry2017ferroelectric}\end{tabular} 
& FeFET & 75 ns & 32 & 0.5\% & 1.79/0.04 $\mu$S & 3.65/-2.95 V \\ \hline
\multirow{2}{*}{\begin{tabular}[c]{@{}c@{}}Ag/PZT/Nb:SrTiO$_3$\\ FTJ \cite{luo2022high}\end{tabular}} & \multirow{2}{*}{FTJ} & 10 ns & 256 & 2.06\% & 80/1 $\mu$S & 1.675/-3.5 V  \\ \cline{3-7} 
&   & 630 ps & 150 & 3.65\% & 27.5/1 $\mu$S &4/-5 V \\ \hline
\end{tabular}
\label{Table: memristors parameters} 
\end{table*}

\section{Discussion}
This work demonstrates the feasibility of UFEE MIMO-OFDM baseband processing by leveraging the emerging in-memory computing technology based on RRAM arrays. The processing latency and energy are mostly contributed by the programming of the RRAM arrays for MIMO detection due to periodic channel estimation, while the following processing of data symbols can be completed in one-step. These advantages promise a feasible approach for realizing UFEE baseband processing. It shall be also emphasized that the proposed in-memory baseband processing not only works on RRAM but can be readily applied to other emerging in-memory computing technologies including phase change, ferroelectric and magnetoresistive memories, as detailed in Table~\ref{Table: memristors parameters} which lists the device features of our experimental RRAM devices and other types of memristor. There are some observations from the simulation results in Fig.~\ref{Fig.11}. First, we compared two different schemes for updating memristor arrays: writing with and without verification, elucidating the importance of verification and low cycle-to-cycle variation in ensuring the accuracy of the operations. To ensure satisfactory communication performance, write-with-verification is suggested for updating the memristor arrays even if the cycle-to-cycle variation is relatively small (e.g., $\sim$0.5$\%$) as shown in the simulation result of programming ferroelectric FET (FeFET)~\cite{jerry2017ferroelectric} without verification. Second, memristor can be further improved using ultra-narrow pulse width along with relatively large number of states to achieve ultra-fast conductance updates without compromising the communication performance. For example, the simulation results in Fig.~\ref{Fig.11} show that the UFEE requirements can be met using ferroelectric tunnel junction (FTJ) which is reported for high precision attainable using sub-nanosecond pulses~\cite{luo2022high}. Leveraging the behavioural model of such memristors, the latency of our in-memory baseband processing system can be reduced to the scale of several microseconds and the energy consumption to the scale of several micro-Jules, which meets the UFEE requirements of 6G communications. Furthermore, in-memory baseband processing is more effective for applications with less stringent precision requirements. For example, if the transmitted messages, such as images, are inputs to the downstream neural networks for inference, the models’ robustness against programming noise can ensure high classification accuracy. Overall, developing the proposed in-memory baseband processing into a versatile technology is believed to provide a feasible approach for realizing the 6G vision on supporting future services and applications with extremely low latency and high energy-efficiency.

\begin{figure}[t!]
    \centering
    \subfigure[MER v.s. SNR]{\label{Fig.5ab_a}
    \includegraphics[width=0.45\columnwidth]{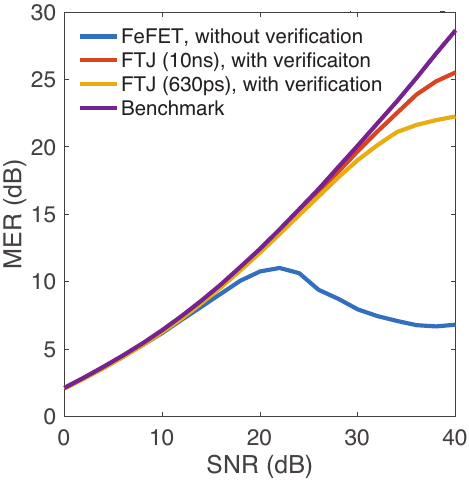}}
    \hspace{2mm}
    \subfigure[BER v.s. SNR]{\label{Fig.5ab_b}
    \includegraphics[width=0.45\columnwidth]{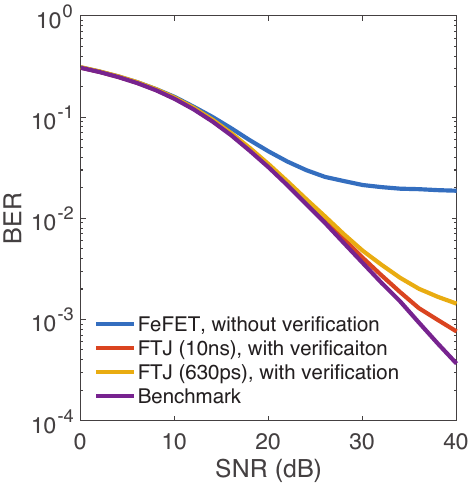}}
    \subfigure[Latency v.s. \#Antennas]{\label{Fig.5cd_a}
    \includegraphics[width=0.45\columnwidth]{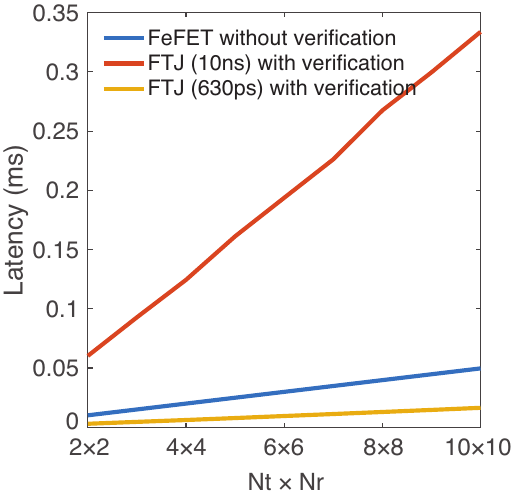}}
    \hspace{2mm}
    \subfigure[Energy v.s. \#Antennas]{\label{Fig.5cd_b}
    \includegraphics[width=0.45\columnwidth]{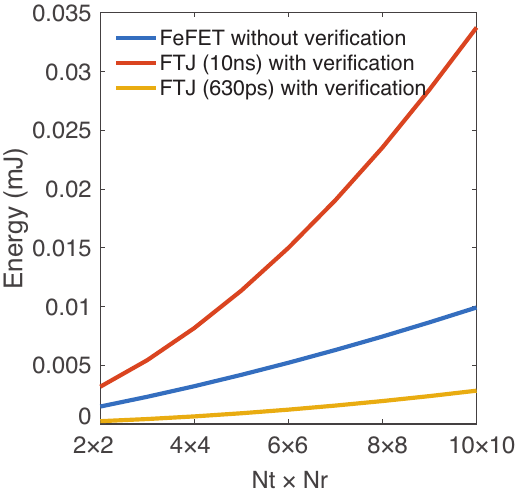}}
    \caption{Performance evaluation of in memory MIMO-OFDM baseband processing using memristors in the literature. For the three memristor behavioural models, the (a) latency and (b) energy are evaluated in terms of different MIMO sizes. Under different channel conditions, the resultant (c) MER (in dB) and (d) BER are shown.}
    \label{Fig.11}
\end{figure}


\appendix

\subsection{RRAM Device Fabrication and Integration}\label{Appendix: RRAM device fabrication}
The integrated chip platform is comprised of three 64 $\times$ 64 Ta/TaO$_x$/Pt RRAM crossbar arrays, together with digital control and analogue sensing circuits to realize in-memory computing. The driving and sensing analogue circuits are taped out with TSMC’s 180nm technology node. After the integration of the RRAM devices with the CMOS circuits, the chip is wire-bonded in a package. The RRAM devices have a lateral dimension of 50nm $\times$ 50nm, fabricated in house with back-end-of-line (BEOL) processes. The layers of the RRAM materials stack (Ta/TaO$_x$/Pt) are deposited with room temperature sputtering, and the electrodes are patterned with electron-beam lithography. The deposited TaO$_x$, serving as the switching layer, has a thickness of $\sim$2nm. 


\subsection{DFT/IDFT Circuit}\label{Appendix: DFT/IDFT and channel estimator circuit}
Consider the circuit with DFT matrix $\mathbf{W}\in\mathbb{C}^{N_c\times N_c}$. The real mapping of the DFT matrix $\mathcal{R}(\mathbf{W})\in\mathbb{R}^{2N_c\times 2N_c}$ is scaled into the RRAM devices’ conductance range by a scaling factor $\alpha$, giving the conductance matrix $\mathbf{G}=\alpha\mathcal{R}(\mathbf{W})\in\mathbb{R}^{2N_c\times 2N_c}$. The real conductance matrix $\mathbf{G}$ is implemented by the difference between a pair of conductance arrays, $\mathbf{G}^+-\mathbf{G}^-$, with the utilization of inverting amplifier to invert the voltages. The received signal $\mathbf{y}\in\mathbb{C}^{N_c\times1}$ is translated to the input voltages with the real vector mapping, such that $\mathbf{v}=\mathcal{T}(\mathbf{y})\in\mathbb{R}^{2N_c\times1}$. Leveraging Ohm’s law (i.e., current = conductance $\times$ voltage), the multiplications $\{G_{kl}^+v_l\}$ and $\{G_{kl}^-v_l\}$ are achieved. Then, Kirchhoff’s current law sums these contributions along each row line and the read circuit integrates all the signals, giving the current at the $k$-th column $i_k=\sum_{l=1}^L(G_{kl}^+-G_{kl}^-)v_l$. Therefore, the output currents at the read circuit give the result: $\mathbf{i}=(\mathbf{G}^+-\mathbf{G}^-)\mathbf{v}$, which gives the DFT result $\alpha\mathcal{T}(\mathbf{x})=\alpha\mathcal{R}(\mathbf{W})\mathcal{T}(\mathbf{y})$. Since DFT matrix is unitary, i.e., $\mathbf{W}^{-1}=\mathbf{W}^{\sf H}$, IDFT module circuit is the same as that of DFT when we replace $\mathcal{R}(\mathbf{W})$ with $\mathcal{R}(\mathbf{W})^{\sf T}$. 

\subsection{L-MMSE/ZF MIMO Detector Circuit}\label{Appendix: L-MMSE/ZF MIMO Detector Circuit}
Consider the L-MMSE detection circuit with channel matrix $\mathbf{H}\in\mathbb{C}^{N_r\times N_t}$. The real mapped channel matrix $\mathcal{R}(\mathbf{H})\in\mathbb{R}^{2N_r\times2N_t}$ is scaled into the RRAM devices’ conductance range by a scaling factor $\alpha$, giving the conductance matrix $\mathbf{G}=\mathbf{G}^+-\mathbf{G}^-=\alpha\mathcal{R}(\mathbf{H})\in\mathbb{R}^{2N_r\times2N_t}$ which is implemented as the difference between two RRAM arrays. The real vector mapping of the received signal $\mathcal{T}(\mathbf{y})\in\mathbb{R}^{2N_r\times1}$ is translated to input currents. To make the voltages in the circuit within a reasonable range, the input currents are also scaled as $\mathbf{i}=\alpha\mathcal{T}(\mathbf{y})\in\mathbb{R}^{2N_r\times1}$. The two arrays at the left-hand side constitute the conductance matrix $-\mathbf{G}=\mathbf{G}^--\mathbf{G}^+$ with voltages $\mathbf{v}$ supplied at the bottom of the nether array. The Kirchhoff’s current law sums the output currents from the left RRAM array pair, $-\mathbf{G}\mathbf{v}$, and the input currents, $\mathbf{i}$, such that the input currents at the operational amplifiers are $\mathbf{i}'=-\mathbf{G}\mathbf{v}+\mathbf{i}$. Hence, the output voltages that supplied to the right RRAM array pair are $\mathbf{v}'=-\frac{\mathbf{i}'}{g_1}=\frac{\mathbf{G}\mathbf{v}-\mathbf{i}}{g_1}$, where $g_1$ is the feedback conductance of the TIAs. Then, the right RRAM array pair, whose conductance matrix is represented by $\mathbf{G}^{\sf T}=(\mathbf{G}^+-\mathbf{G}^-)^{\sf T}$, performs the MVM computation and outputs the current vector $\mathbf{i}''=\mathbf{G}^{\sf T}\mathbf{v}'=\mathbf{G}^{\sf T}\frac{\mathbf{G}\mathbf{v}-\mathbf{i}}{g_1}$. The currents are applied to the other set of TIAs, so that $\mathbf{i}''=-g_2\mathbf{v}$, where $g_2$ is the feedback conductance of the TIAs in this set. Accordingly, one can observe the relation: $\mathbf{G}^{\sf T}\frac{\mathbf{G}\mathbf{v}-\mathbf{i}}{g_1}=-g_2\mathbf{v}$, which gives the output voltages $\mathbf{v}=\left(\mathbf{G}^{\sf T}\mathbf{G}+g_1 g_2\mathbf{I}\right)^{-1}\mathbf{G}^{\sf T}\mathbf{i}$. By setting the SNR as $\alpha^2(g_1 g_2 )^{-1}$, the designed L-MMSE circuit outputs the desired vector: $\mathcal{T}(\hat{\mathbf{x}})=\left(\mathcal{R}(\mathbf{H})^{\sf T}\mathcal{R}(\mathbf{H})+\frac{1}{\sf SNR}\mathbf{I}_{2N_t\times 2N_t} \right)^{-1}\mathcal{R}(\mathbf{H})^{\sf T}\mathcal{T}(\mathbf{y})$. When the feedbacks of the TIAs in the second set are open, i.e., $g_2=0$, the output voltages of the circuit are $\mathbf{v}=(\mathbf{G}^{\sf T}\mathbf{G})^{-1}\mathbf{G}^{\sf T}\mathbf{i}$. This computes the ZF and gives the desired vector $\mathcal{T}(\hat{\mathbf{x}})=\left(\mathcal{R}(\mathbf{H})^{\sf T}\mathcal{R}(\mathbf{H})\right)^{-1}\mathcal{R}(\mathbf{H})^{\sf T}\mathcal{T}(\mathbf{y})$.

 \subsection{Implementing Complex Matrices and Vectors}\label{Appendix: implementing complex matrices and vectors}
Both the channel entries and signals are in the complex domain while employing differential pairs of RRAM devices can only represent real numbers. To address this issue, we propose to apply the equivalent matrices and vectors of real entries instead. Inspired by the matrix representation of complex numbers, i.e., the mapping $a+b\mathrm{j}\to\left(\begin{matrix}a&-b\\b&a\end{matrix}\right)$ is a ring isomorphism from the field of complex numbers to the ring of these matrices, we extend the method to complex matrices and define the mappings as follows.
\begin{definition}\label{definition: real matrix mapping}
(Real Matrix Mapping). \emph{Define the mapping $\cR:\mC^{K\times L}\to\mR^{2K\times2L}$, which transforms a complex matrix $\bA=\Re(\bA)+{\rm j}\Im(\bA)\in\mC^{K\times L}$ into a real matrix $\cR(\bA)\in\mR^{2K\times2L}$:
\begin{equation}
    \cR(\bA)=\l[\begin{matrix}
    \Re(\bA)&-\Im(\bA)\\
    \Im(\bA)&\Re(\bA)
    \end{matrix}\r].
 \end{equation}
 }
 \end{definition}
 The defined mapping preserves the basic operations of matrices (see Lemma~\ref{lemma: real mapping properties}), making it a feasible method for in-memory baseband processing implementation.
\begin{lemma}\label{lemma: real mapping properties}
(Properties of Equivalent Real Matrices). \emph{Some basic properties of the mapping $\cR$ defined in Definition \ref{definition: real matrix mapping} are described as follows. For any matrices $\bA,\bB\in\mC^{K\times L}$, $\bC\in\mC^{L\times N}$,
 \begin{align}
    \cR(\bA)+\cR(\bB)&=\cR(\bA+\bB),\\
    \cR(\bA)\cR(\bC)&=\cR(\bA\bC),\label{eqn: lemma 2}\\
    \cR(\bA^{\sf H})&=\cR(\bA)^{\sf T}.
\end{align}
 }
\end{lemma}
The proof involves straightforward calculations of matrices and thus omitted for brevity.
It can be inferred from the equation \eqref{eqn: lemma 2} that $\cR(\bA^{-1})=\cR(\bA)^{-1}$ if $\bA$ is invertible. Given the mapping, the complex matrices can be written into the RRAM arrays without specific changes or auxiliary circuits. 
On the other hand, the proposed method can be applied to complex vectors as well, where one complex vector $\bx\in\mC^{K\times1}$ is mapped to a real matrix $\cR(\bx)\in\mR^{2K\times2}$. Then the \emph{matrix-vector multiplication} (MVM) can be achieved following equation \eqref{eqn: lemma 2}. However, it takes two steps to complete the operation since $\cR(\bx)$ is a matrix with two columns. To further improve the computational efficiency, we propose to implement the complex vector, which is usually the input voltages/currents for RRAM array, using the following transformation.
\begin{definition}\label{definition: real vector mapping}
(Real Vector Mapping). \emph{Define the mapping $\cT:\mC^{K\times 1}\to\mR^{2K\times1}$, which transforms a complex vector $\bx\in\mC^{K\times L}$ into a real vector $\cT(\bx)\in\mR^{2K\times1}$:
\begin{equation}
    \cT(\bx)=\l(\begin{matrix}
    \Re{(\bx})\\
    \Im{(\bx)}
    \end{matrix}\r).
\end{equation}
}
\end{definition}
The proposed mappings in Definition \ref{definition: real matrix mapping} and \ref{definition: real vector mapping} make it possible to realize one-shot MVM computation as shown below.
 \begin{lemma}\label{lemma: MVM real matrix and vector}
 (One-Shot MVM Operation Between Equivalent Real Matrix and Vector). \emph{For any matrix $\bA\in\mC^{K\times L}$ stored in RRAM array and vector $\bx\in\mC^{L\times1}$ translated as supply voltages, one-shot MVM is realized by the relation:
 \begin{equation}
    \cT(\bA\bx)=\cR(\bA)\cT(\bx).
 \end{equation}
 }
 \end{lemma}
 It is proved by checking the following two relations: $$\Re(\bA\bx)=\Re(\bA)\Re(\bx)-\Im(\bA)\Im(\bx)$$ and $$\Im(\bA\bx)=\Im(\bA)\Re(\bx)+\Re(\bA)\Im(\bx).$$ 
 Based on Lemmas \ref{lemma: real mapping properties} and \ref{lemma: MVM real matrix and vector}, the following two useful equations can be obtained:
 \begin{equation}
    \cT\l((\bA^{\sf H}\bA)^{-1}\bA^{\sf H}\bx\r)=\l(\cR(\bA)^{\sf T}\cR(\bA)\r)^{-1}\cR(\bA)^{\sf T}\cT(\bx).
 \end{equation}
 \begin{equation}
    (\mathbf{A}^{\sf H}\mathbf{A}+\lambda\mathbf{I})^{-1}\mathbf{A}^{\sf H}\mathbf{x}=\l(\cR(\bA)^{\sf T}\cR(\bA)+\lambda\mathbf{I}\r)^{-1}\cR(\bA)^{\sf T}\cT(\bx).
 \end{equation}
where $\lambda$ is a constant. The above equations correspond to MIMO detection implementation.

\subsection{Latency Analysis}\label{Appendix: latency analysis}
We aim at quantifying the latency of writing a MIMO channel matrix into a RRAM array in the row-by-row manner. In particular, the channel is assumed as Rayleigh fading while the RRAM array is comprised of 1T1R cells. Two writing schemes are analyzed: write-without-verification and write-with-verification.
\subsubsection{Communication model}
The input-output relation of a MIMO system with channel matrix $\mathbf{H}\in\mathbb{C}^{N_r\times N_t}$ is described as
 \begin{equation}
    \by=\bH\bx+\bz,
 \end{equation}
where $\bx\in\mC^{N_t\times1}$ and $\by\in\mC^{N_r\times1}$ denote the transmit and receive symbols, respectively. $\bz\sim\cC\cN(0,\sigma_n^2 \bI_{N_r})$ represents the AWGN in propagation. We consider i.i.d. Rayleigh fading model, where the entries of the channel matrix $\mathbf{H}$ follow i.i.d. zero-mean complex Gaussian distribution. For any $(i,j)$-th element in matrix $\bH$, we have $H_{ij}\sim\cC\cN(0,2\sigma_h^2)$ where $\Re(H_{ij})\sim\cN(0,\sigma_h^2)$ and $\Im(H_{ij})\sim\cN(0,\sigma_h^2)$.

The elements of channel matrix $\{H_{ij}\}$ vary in the whole real domain. We need to scale them into the feasible conductance range for RRAM devices. For each RRAM device, the maximum and minimum values of working conductance are denoted by $G_{\max}$ and $G_{\min}$, respectively. For convenience, we assume $G_{\min}=0$ such that a differential pair of RRAM devices can represent a real number in the interval $[-G_{\max},G_{\max}]$. In order to scale the channel matrix into this range, we apply the three-sigma rule:
\begin{align}
    &\Pr(-3\sigma_h\leq\Re{(H_{ij})}\leq3\sigma_h)=99.73\%, \\ &\Pr(-3\sigma_h\leq\Im{(H_{ij})}\leq3\sigma_h)=99.73\%.
\end{align}
Through this rule, we can guarantee the feasibility of representing channel matrix by a RRAM array with high probability. Accordingly, the variance of the channel elements as mentioned becomes $\sigma_h=G_{\max}/3$.


\subsubsection{Model of RRAM device}
To begin with, we focus on the RRAM device whose behavioral model is shown in Fig.~\ref{Fig.2}(f) and Fig.~\ref{Fig.2}(g). Consider the conductance update process using a train of write pulses. The pulse width is denoted as $\Delta t_w$ which is a minuscule value. Cycle-to-cycle variation $\sigma_c$ refers to the variation in conductance change at every write pulse. It is expressed as the percentage of the entire conductance range in the existing literature \cite{yu2018neuro,chen2018neurosim,chen2018technological,sun2019impact}. In other words, the standard variance of the per-cycle write noise, $\cN(0,\sigma_c^2)$, is presented by $\sigma_c=\gamma G_{\max}$ where $\gamma\in(0,1)$ denotes the percentage. Then, we can characterize the per-pulse conductance change as follows:
\begin{equation}\label{eqn: one cycle update}
    \Delta G=\frac{G_{\max}-G_{\min}}{N_p}+\frac{\sigma_c}{\sqrt{\Delta t_w}}\Delta W,
\end{equation}
where $\Delta G=G(t+\Delta t_w)-G(t)$ is the conductance change by applying one write pulse over $G(t)$, $N_p$ is pulse number corresponding to programming conductance from $G_{\min}$ to $G_{\max}$, and $\Delta W=W(t+\Delta t_w)-W(t)$ with $W(t)$ being a Winner process: $\Delta W\sim\cN(0,\Delta t_w)$. From \eqref{eqn: one cycle update}, we know the following knowledge of the RRAM device's state after one pulse:
\begin{align}
    \mE[G(t+\Delta t_w)| G(t)]&=G(t)+\frac{G_{\max}-G_{\min}}{N_p\times\Delta t_w}\Delta t_w,\\
    \var[G(t+\Delta t_w)| G(t)]&=\left(\frac{\sigma_{\sf c}}{\sqrt{\Delta t_w}}\right)^2\Delta t_w.
\end{align}
By introducing the slope parameter $\mu\triangleq\frac{G_{\max}-G_{\min}}{N_p\times \Delta t_w}$ and the diverting variance $\sigma\triangleq\frac{\sigma_c}{\sqrt{\Delta t_w}}$, the evolution of conductance state $G(t)$ is characterized by the stochastic differential equation (SDE) given the initial state $G_0$ (i.e., the conductance at the time $t=0$):
\begin{equation}\label{eqn: linear SDE}
    dG(t)=\mu dt+\sigma dW(t),\quad 
    G(0)=G_0.
\end{equation}
The solution of \eqref{eqn: linear SDE} gives an It\^o process: 
\begin{equation}
    G(t)=G_0+\mu t+\sigma\int_0^tdW(s).
\end{equation}
    
\subsubsection{Performance of RRAM device}
Let $p(g,t|G_0)$ represent the conditional probability density of $G(t)$ given initial state $G(0)=G_0$. For the SDE specified in \eqref{eqn: linear SDE}, the probability density of the solution satisfies the forward Kolmogorov equation (also known as Fokker-Planck equation) with the initial condition as follows:
\begin{equation}
\begin{aligned}\label{eqn: PDE linear model}
    \frac{\partial p(t,g)}{\partial t}&=-\mu\frac{\partial p(t,g)}{\partial g}+\frac{\sigma^2}{2}\frac{\partial^2p(t,g)}{\partial g^2},\\
    p(0,g)&=\delta(g-G_0),
\end{aligned}
\end{equation}
where $\delta(\cdot)$ is Dirac function. By solving the partial differential equation in \eqref{eqn: PDE linear model}, we obtain the following lemma.
\begin{lemma}\label{lemma: linear model probability density}
(Probability Density). \emph{The probability density for the conductance evolution $G(t)$ is
\begin{equation}
    p(t,g|G_0)=\frac{1}{\sqrt{2\pi\sigma^2t}}\exp\l(-\frac{(g-\mu t-G_0)^2}{2\sigma^2t}\r),
\end{equation}
which gives the Gaussian distribution $(G(t)|G_0)\sim\cN(G_0+\mu t,\sigma^2)$.}
\end{lemma}
The process $\{G(t),t\geq0\}$ is time homogeneous with independent increments. Without loss of generality, we assume the target conductance $\bar{G}$ is larger than the initial state, i.e., the increment of conductance is positive: $\Delta G\triangleq \bar{G}-G_0\geq0$. The writing time is denoted as $T$, which aims at increasing the conductance by the amplitude of $\Delta G$.
\begin{itemize}
    \item \textbf{Write-without-verification scheme}
    
    Given target conductance $\bar{G}$ and the increment $\Delta G$, the write latency of the RRAM device is determined by
    \begin{equation}
        T=\frac{\Delta G}{\mu}.
    \end{equation}
    Meanwhile, the achieved conductance state $G(T)$ is inaccurate, giving that
    \begin{equation}
        G(T)\sim\cN\Big(\bar{G},\frac{\sigma^2\Delta G}{\mu}\Big).
    \end{equation}
    \item \textbf{Write-with-verification scheme}
    
    To achieve the target conductance $\bar{G}$, the read pulse is applied after each write pulse to monitor the evolution of conductance state $G(t)$. We model it as the first passage time, which refers to the first time when the conductance state $G(t)$ achieves the target value $\bar{G}$,
    \begin{equation}
        T\triangleq\inf\{t\geq 0: G(t)=\bar{G}\}.
    \end{equation}
    By adding an absorbing boundary $p(t,\bar{G})=0$ to the partial differential equation (together with the initial condition) in \eqref{eqn: PDE linear model}, we obtain the probability density of the first passage time as the solution of the boundary value problem.
    \begin{lemma}
    (First Passage Time Probability Density). \emph{The probability density of the first passage time with the target increment conductance $\Delta G$ is given by
    \begin{equation}
        p(T|\Delta G)=\frac{\Delta G}{\sqrt{2\pi\sigma^2T^3}}\exp\l(-\frac{(\Delta G-\mu T)^2}{2\sigma^2T}\r),
    \end{equation}
    which gives the inverse Gaussian distribution $(T|\Delta G)\sim\cI\cG\Big(\Delta G/\mu,(\Delta G/\sigma)^2\Big)$.}
\end{lemma}
\end{itemize}

\subsubsection{Latency of RRAM array programming}
We consider the latency of writing the real mapped channel matrix $\cR(\hat{\bH})$ onto the 1T1R array in the row-by-row manner. For simplicity, we assume the array has been fully reset, i.e., all the RRAM devices are initialized with $G_{ij}(0)=0,~\forall(i,j)$. In this setting, $G_{ij}^+$ will be updated if $H_{ij}\geq0$, or $G_{ij}^-$ will be updated otherwise.
The conductance change of one device, e.g., $\Delta G_{ij}$ for the $(i,j)$-th device, follows the i.i.d. half-normal distribution over indices $\forall (i,j)$, giving that 
\begin{equation}
    \Delta G_{ij}\sim\l|\cN\l(0,{G_{\max}^2}/{9}\r)\r|,\quad i=1,\cdots,N_t,~j=1,\cdots,N_r.
\end{equation}
The matrix $\cR(\bH)$ has $2N_r$ rows and it consists of two identical processes each with updating $N_r$ rows. Thus, the expected latency of writing this matrix is
\begin{align}\label{eqn: write one array vs one row}
    T_{\sf write}=2\times\mE\l[\sum_{i=1}^{N_r}T_i^{\sf row}\r]=2N_r\times\mE[T_i^{\sf row}],
\end{align}
where the expectation is taken over channel entries. One-row latency, say the $i$-th row, is determined by the RRAM device consuming the largest write time, that is
\begin{equation}
    T_i^{\sf row}=\max_{1\leq j\leq 2N_t}T_{ij},
\end{equation}
where latencies $\{T_{ij}\}$ refer to writing $\Re{(H_{ij'})}$ and $\Im{(H_{ij'})}$, $j'=1,\cdots,N_t$ into the $i$-th row of RRAM array. Hence, latencies $\{T_{ij}\}_{j=1}^{2N_t}$ are $2N_t$ i.i.d. random variables. 
\begin{itemize}
    \item \textbf{Write-without-verification scheme}
    
    Recall that the write time of one RRAM device using write-without-verification scheme is determined by $T=\Delta G/\mu$. Thus, for the $i$-th row, the write time of the RRAM device at the $j$-th column follows the half-normal distribution, that is 
    \begin{equation}
        T_{ij}\sim\l|\cN\l(0,\frac{G_{\max}^2}{9\mu^2}\r)\r|,\quad j=1,\cdots,2N_r.
    \end{equation}
    where $|\cN(\cdot,\cdot)|$ denotes the half-normal distribution.
    \begin{theorem}
    (Expected Latency of Write-Without-Verification). \emph{Consider writing the scaled real mapped channel matrix into an 1T1R array row-by-row using write-without-verification scheme. The expected write latency is upper bounded by
    \begin{equation}
        T_{\sf write}\leq\frac{2\sqrt{2}G_{\max}}{3\mu}N_r\l(\sqrt{\ln{N_t}}+\frac{1}{\sqrt{\pi}\ln{N_t}}\r).
    \end{equation}
    }
    \end{theorem}
    \begin{proof}
    Since $T_i^{\sf row}=\max_{1\leq j\leq 2N_t}T_{ij}$, the probability of $T_i^{\sf row}$ satisfies
    \begin{align}
        \Pr\big(T_i^{\sf row}> t\big)&=1-\Pr\l(\max_{1\leq j\leq 2N_t}T_{ij}\leq t\r)\nn\\
        &=1-\big(1-\Pr(T_{ij}> t)\big)^{2N_t}.
    \end{align}
    Applying Bernoulli's inequality, we know that 
    \begin{equation}
    	\l(1-\Pr(T_{ij}> t)\r)^{2N_t}\geq 1-2N_t\Pr(T_{ij}> t).
    \end{equation}
    Thus, for the non-negative random variable $T_i^{\sf row}$, its expectation is expressed as
    \begin{align}
        \mE[T_i^{\sf row}]&=\int_0^{\infty}\Big(1-\big(1-\Pr(T_{ij}> t)\big)^{2N_t}\Big)dt\nn\\
        &\leq \varepsilon+\int_{\varepsilon}^{\infty}\Big(1-\big(1-\Pr(T_{ij}> t)\big)^{2N_t}\Big)dt\nn\\
        &\leq\varepsilon+2N_t\int_{\varepsilon}^{\infty}\Pr(T_{ij}>t)dt,
    \end{align}
    where the inequality holds for any positive constant $\varepsilon>0$. To obtain the probability $\Pr(T_{ij}>t)$, we introduce the \emph{cumulative distribution function} (CDF) of $T_{ij}$, i.e., 
        $F_{T_{ij}}(t)=\text{erf}\l(\frac{3\mu t}{\sqrt{2}G_{\max}}\r)$.
    To ease the notation, we denote the parameter $\bar{\sigma}\triangleq\frac{G_{\max}}{3\mu}$. Then the probability is given by 
    \begin{equation}
    	\Pr(T_{ij}>t)=1-F_{T_{ij}}(t)=2Q(t/\bar{\sigma}),
    \end{equation}
    where $Q(\cdot)$ is the Q-function. Leveraging the inequality $Q(x)\geq\frac{1}{x\sqrt{2\pi}}\l(1-\frac{1}{x^2}\r)e^{-x^2/2}$ for $x>0$, we have
    \begin{align}
        \mE[T_i^{\sf row}]
        &\leq\varepsilon+4N_t\int_{\varepsilon}^{\infty}Q\l(\frac{t}{\bar{\sigma}}\r)dt\nn\\
        &=\varepsilon+2\sqrt{\frac{2}{\pi}}N_t\bar{\sigma}\exp\l(-\frac{\varepsilon^2}{2\bar{\sigma}^2}\r)-4N_t\varepsilon Q\l(\frac{\varepsilon}{\bar{\sigma}}\r)\nn\\
        &\leq \varepsilon+2\sqrt{\frac{2}{\pi}}N_t\frac{\bar{\sigma}^3}{\varepsilon^2}\exp\l(-\frac{\varepsilon^2}{2\bar{\sigma}^2}\r).
    \end{align}
    Substituting $\varepsilon=\bar{\sigma}\sqrt{2\ln{N_t}}$ into the inequality for $N_t>1$, we obtain 
    \begin{align}
        \mE[T_i^{\sf row}]&\leq \bar{\sigma}\sqrt{2\ln{N_t}}+\bar{\sigma}\sqrt{\frac{2}{\pi}}\frac{1}{\ln{N_t}}\nn\\
        &=\frac{\sqrt{2}G_{\max}}{3\mu}\l(\sqrt{\ln{N_t}}+\frac{1}{\sqrt{\pi}\ln{N_t}}\r).
    \end{align}
    Finally, according to \eqref{eqn: write one array vs one row}, the expected latency of updating the whole RRAM array is upper bounded by
    \begin{equation}
        T_{\sf write}\leq\frac{2\sqrt{2}G_{\max}}{3\mu}N_r\l(\sqrt{\ln{N_t}}+\frac{1}{\sqrt{\pi}\ln{N_t}}\r).
    \end{equation}
    This completes the proof.
    \end{proof}

    \item \textbf{Write-with-verification scheme}
    
    Recall that the latency of writing one RRAM device, say the $(i,j)$-th device, in this scheme is determined by the first passage time with the compound distribution as follows:
    \begin{equation}
    \begin{aligned}
        T_{ij}|\Delta G_{ij}&\sim\cI\cG\l({\Delta G_{ij}}/{\mu},{(\Delta G_{ij})^2}/{\sigma^2}\r),\\
        \Delta G_{ij}&\sim\l|\cN\l(0,{G_{\max}^2}/{9}\r)\r|.
    \end{aligned}
    \end{equation}
    where $\cI\cG(\cdot,\cdot)$ and $|\cN(\cdot,\cdot)|$ denote the inverse Gaussian distribution and half-normal distribution, respectively.
    \begin{theorem}
    (Expected Latency of Write-With-Verification). \emph{Consider writing the scaled real mapped channel matrix into an 1T1R array row-by-row using write-with-verification scheme. The expected write latency is upper bounded by
    \begin{equation}
    \begin{aligned}
        T_{\sf write}\leq 2N_r\times\min&\l\{\frac{2\sqrt{2}G_{\max}}{3\mu}\sqrt{\ln(4N_t)},\r.\\
        &~\l.\frac{2\sigma^2}{\mu^2}\ln(4N_t)+\frac{G_{\max}^2}{9\sigma^2}\r\}.
    \end{aligned}
    \end{equation}
    }
    \end{theorem}
    \begin{proof}
    Recall that $T_i^{\sf row}=\max_{1\leq j\leq 2N_t}T_{ij}$ with i.i.d. $\{T_{ij}\}_{j=1}^{2N_t}$. Applying Jensen's inequality, 
    \begin{align}
        \exp(\varepsilon\mE[T_i^{\sf row}])&\leq\mE[\exp(\varepsilon T_i^{\sf row})]\nn\\
        &=\mE\l[\max_{1\leq j\leq 2N_t}\exp(\varepsilon T_{ij})\r]\nn\\
        &\leq \sum_{j=1}^{2N_t}\mE\l[\exp(\varepsilon T_{ij})\r]\nn\\
        &=2N_t\mE\l[\exp(\varepsilon T_{ij})\r],
    \end{align}
    where the expectation follows the rule of compound distribution, that is
    \begin{equation}
        \mE\l[\exp(\varepsilon T_{ij})\r]=\mE_{\Delta G_{ij}}\big[\mE\l[\exp(\varepsilon T_{ij})|\Delta G_{ij}\r]\big].
    \end{equation}
    Given $\Delta G_{ij}$, the latency $T_{ij}$ follows an inverse Gaussian distribution, $T_{ij}\sim\cI\cG\l(\frac{\Delta G_{ij}}{\mu}, \frac{(\Delta G_{ij})^2}{\sigma^2}\r)$, whose \emph{moment-generating function} (MGF) is 
    \begin{equation}
        \mE\l[\exp(\varepsilon T_{ij})|\Delta G_{ij}\r]=\exp\!\l\{\frac{\mu\Delta G_{ij}}{\sigma^2}\!\l(1-\sqrt{1-\frac{2\sigma^2\varepsilon}{\mu^2}}\r)\!\r\}.
    \end{equation}
    To ease the notation, we define 
    \begin{equation}
    	x(\varepsilon)\triangleq\frac{\mu}{\sigma^2}\l(1-\sqrt{1-\frac{2\sigma^2\varepsilon}{\mu^2}}\r),
    \end{equation}
    which is an increasing function for $\varepsilon\in(0,\frac{\mu^2}{2\sigma^2}]$. Furthermore, we denote the standard variance of $\Delta G_{ij}$ as $\tilde{\sigma}$ whose value is $\tilde{\sigma}=\frac{G_{\max}}{3}$. Then, leveraging the MGF of the half-normal distribution, $\Delta G_{ij}\sim|\cN(0,\tilde{\sigma}^2)|$, we have
    \begin{align}
        \mE[\exp(x(\varepsilon)\Delta G_{ij})]=\exp\l(\frac{\tilde{\sigma}^2x(\varepsilon)^2}{2}\r)\l(1+{\rm erf}\l(\frac{\tilde{\sigma} x(\varepsilon)}{\sqrt{2}}\r)\r).
    \end{align}
    From the relation $\exp(\varepsilon\mE[T_i^{\sf row}])\leq2N_t\mE[\exp(\varepsilon T_{ij})]=2N_t\mE[\exp(x(\varepsilon)\Delta G_{ij})]$, we have
    \begin{align}
        \mE[T_i^{\sf row}]\leq &\frac{1}{\varepsilon}\ln{(2N_t)}+\frac{\tilde{\sigma}^2x(\varepsilon)^2}{2\varepsilon}\nn\\
        &+\frac{1}{\varepsilon}\ln\l(1+{\rm erf}\l(\frac{\tilde{\sigma}x(\varepsilon)}{\sqrt{2}}\r)\r).
    \end{align}
    It is not hard to prove the following facts: $x(\varepsilon)\leq\frac{2}{\mu}\varepsilon$ and ${\rm erf}(\cdot)\leq 1$, and thus we have
    \begin{align}\label{eqn: R58}
        \mE[T_i^{\sf row}]\leq \frac{1}{\varepsilon}\ln{(4N_t)}+\frac{2\tilde{\sigma}^2}{\mu^2}\varepsilon,
    \end{align}
    which holds for any $\varepsilon\in(0,\frac{\mu^2}{2\sigma^2}]$. This indicates the expected one-row latency $\mE[T_i^{\sf row}]$ is upper bounded by the minimum of the right-hand side of inequality \eqref{eqn: R58}. 
    
    1). If $\ln(4N_t)\leq\frac{\mu^2}{2\sigma^2}\l(\frac{\tilde{\sigma}}{\sigma}\r)^2$, the minimum value of the right-hand side of \eqref{eqn: R58} is
    \begin{align}
        \frac{1}{\varepsilon}\ln{(4N_t)}+\frac{2\tilde{\sigma}^2}{\mu^2}\varepsilon
        &\geq 2\sqrt{\frac{1}{\varepsilon}\ln{(4N_t)}\cdot\frac{2\tilde{\sigma}^2}{\mu^2}\varepsilon}\nn\\
        &=\frac{2\sqrt{2}G_{\max}}{3\mu}\sqrt{\ln(4N_t)},
    \end{align}
    and thus, we have 
    \begin{equation}
        \mE[T_i^{\sf row}]\leq\frac{2\sqrt{2}G_{\max}}{3\mu}\sqrt{\ln(4N_t)}.
    \end{equation}
    
    2). If $\ln(4N_t)>\frac{\mu^2}{2\sigma^2}\l(\frac{\tilde{\sigma}}{\sigma}\r)^2$, the minimum value of the right-hand side of inequality \eqref{eqn: R58} is achieved at $\epsilon=\frac{\mu^2}{2\sigma^2}$, so that 
    \begin{align}
        \frac{1}{\varepsilon}\ln{(4N_t)}+\frac{2\tilde{\sigma}^2}{\mu^2}\varepsilon
        &\geq \frac{2\sigma^2}{\mu^2}\ln(4N_t)+\l(\frac{\tilde{\sigma}}{\sigma}\r)^2\nn\\
        &=\frac{2\sigma^2}{\mu^2}\ln(4N_t)+\frac{G_{\max}^2}{9\sigma^2},
    \end{align}
    and thus, we have
    \begin{equation}
        \mE[T_i^{\sf row}]\leq\frac{2\sigma^2}{\mu^2}\ln(4N_t)+\frac{G_{\max}^2}{9\sigma^2}.
    \end{equation}
    
    By combining the results in the two cases, we have
    \begin{equation}
    \begin{aligned}
        &\mE\!\l[T_i^{\sf row}\r]\leq\\
        &\begin{cases}
        \frac{2\sqrt{2}G_{\max}}{3\mu}\sqrt{\ln(4N_t)},&\ln(4N_t)\leq\frac{\mu^2}{2\sigma^2}\l(\frac{G_{\max}}{3\sigma}\r)^2,\\ 
        \frac{2\sigma^2}{\mu^2}\ln(4N_t)+\frac{G_{\max}^2}{9\sigma^2},&\ln(4N_t)>\frac{\mu^2}{2\sigma^2}\l(\frac{G_{\max}}{3\sigma}\r)^2,
        \end{cases}        
    \end{aligned}
    \end{equation}
    which can be easily verified that it is equivalent to 
    \begin{equation}
    \begin{aligned}
        \mE\l[T_i^{\sf row}\r]\leq \min&\l\{\frac{2\sqrt{2}G_{\max}}{3\mu}\sqrt{\ln(4N_t)},\r.\\
        &~\l.\frac{2\sigma^2}{\mu^2}\ln(4N_t)+\frac{G_{\max}^2}{9\sigma^2}\r\}.
    \end{aligned}
    \end{equation}
    Finally, according to \eqref{eqn: write one array vs one row}, the expected latency of updating the whole RRAM array is upper bounded by
    \begin{equation}
    \begin{aligned}
        T_{\sf write}\leq 2N_r\times\min&\l\{\frac{2\sqrt{2}G_{\max}}{3\mu}\sqrt{\ln(4N_t)},\r.\\
        &~\l.\frac{2\sigma^2}{\mu^2}\ln(4N_t)+\frac{G_{\max}^2}{9\sigma^2}\r\}.
    \end{aligned}
    \end{equation}
    This completes the proof.
    \end{proof}
\end{itemize}

\subsection{Comparison with Digital CMOS Counterpart}
In this note, we provide the comparisons with the state-of-the-art (SOTA) digital CMOS counterparts, which is important to highlight the advantages of our RRAM-based baseband processor in terms of latency (or throughput) and energy efficiency. The calculations are based on the parameters in the Table \ref{table: comparison CMOS} unless specified otherwise.
\begin{table*}[t!]
\centering
\caption{Parameters for Comparison with CMOS-based Digital Processors}\label{table: comparison CMOS}
\begin{tabular}{|c|c|c|c|c|c|}
\hline
Parameter & DFT Size                                             & \# Slots/Frame & Symbols/Slot & Modulation & Number of Antennas                                                           \\ \hline
Value     & \begin{tabular}[c]{@{}c@{}}1024\\ 2048*\end{tabular} & 160            & 14           & 16-QAM     & \begin{tabular}[c]{@{}c@{}}4$\times$4\\ 2$\times$2$+$4$\times$4*\end{tabular} \\ \hline
\end{tabular}\\ \vspace{2mm}
*2048 and 2 $\times$ 2 + 4 $\times$ 4 are only used for comparison with Qualcomm Snapdragon X65.
\end{table*}

\subsubsection{Comparison with SOTA commercial modem}
SOTA commercial modems are fabricated using the latest technology node, targeting the 5G signal processing. Take the popular Qualcomm Snapdragon X65 modem as an example. It is fabricated using the TSMC 4nm process according to the public data \cite{snapdragonx65}. The maximum download speed can achieve up to 10 Gb/s over mmWave and sub-6 carrier aggregation. Note that the modulation used in this product is 256-QAM (8 bits/symbol) while we used 16-QAM (4 bits/symbol) in this work for the demonstration of our RRAM-based design. Therefore, the top speed is reduced by half as we match the QAM order of the two processors for a fair comparison, resulting in a peak throughput of 5 Gb/s as benchmark. In addition, the throughput of X65 modem is the combination of two separate bands and specifications, i.e., mmWave (2 $\times$ 2 MIMO) and sub-6 GHz (4 $\times$ 4 MIMO). For the same communication overhead, the baseband processing latency of our design with the same parameters (2048-point DFT, 2 $\times$ 2 MIMO + 4 $\times$ 4 MIMO) is only 0.2409 ms by re-simulation. It means our design can support a communication throughput as large as 455.8 Gb/s, which is 91.2$\times$ higher than that of the SOTA (i.e., Snapdragon X65 modem). However, it is hard to scale the throughput to fit the other group of settings specified in Table 1 which are used for other baselines, so that we just summarize it as $<$ 5 Gb/s in Table II, making it appear different from the values in other designs.  
On the other hand, the energy efficiency comparison is not feasible since the X65 modem is an integrated system-on-chip (SoC) that contains RFIC, control processor, digital baseband processor, and other units. Unfortunately, there is no detailed energy efficiency data related to baseband processing in this modem available in the public domain.
\subsubsection{Comparison with multi-core digital signal processor (DSP) reported in literature}
Next, we compare the proposed design with the reported powerful multi-core DSP in the literature, namely TMS320C6678 from Texas Instruments \cite{fryza2014power}. The performance achieves up to 128 GOPS while the average power consumption for the optimized digital processing function is 11.1 Watts. For fair comparison, we just take the power of this 8-core DSP into consideration while neglect the power of the development board TMDSEVM6678LE. The baseband processing workload is estimated by the algorithmic computation complexity, giving the value of 75.5 GOPs. Accordingly, the latency and energy consumption are calculated as 589.9 ms and 6.548 J, respectively. Meanwhile, ignoring the power of peripheral circuits, our RRAM-based design shows the latency and energy consumption being 0.2278 ms and 0.0079 mJ, and thus the equivalent computational throughput and energy efficiency are calculated as 331.4 TOPS and 9557 TOPS/Watt which outperforms this DSP by $10^3$ and $10^5$ times, respectively.
\subsubsection{Comparison with digital baseband processor reported in literature}
Moreover, we compare the proposed design with the reported digital baseband processor in the literature. Since the design of traditional digital baseband processor is considered a matured area, there are few recent publications on the complete baseband system design. As a compromise, we provide the following two comparisons: a) one with the SOTA domain adaptive processor (DAP) for wireless communication; b) the other with “virtually assembled” digital baseband processor by combining the SOTA designs of isolated digital baseband modules (i.e., DFT, MIMO detection, etc.) collected from the recent literatures.
\begin{itemize}
    \item[a) ] \emph{Comparison with SOTA digital adaptive processor (DAP):}
    The accelerator presented in the reference \cite{chen2022507} has been fabricated by a 12nm technology node and specialized for wireless communication workloads. Its peak performance reaches 264 GOPS at power consumption of 272 GOPS/Watt. One can observe it has a noticeable performance gain compared with commercial DSPs \cite{fryza2014power}. However, the algorithms for computing FFT and MMSE in this design are different from the discussions in Supplementary Note 9, so it is not reasonable to directly compare the computational performance. Hence, we evaluate the performance of the baseline in terms of communication throughput and energy efficiency in the way as follows. The measurement results of this DAP show the throughput and energy efficiency are 4.41G samples/s and 53.96 nJ/FFT for 256-point FFT. For a fair comparison, we consider the joint utilization of 16 DAPs to complete the 1024-point DFT by decomposing this large-scale DFT into 16 small-scale 256-point FFTs. Accordingly, the latency is saved at the cost of more area and energy consumption. On the other hand, the results from the measurements on MMSE MIMO detection reveal the throughput and energy efficiency being 1.95M matrices/s and 178.5 nJ/matrix, respectively. Moreover, there are additional reprogram times in OFDM demodulation (0.5$\mu$s) and MIMO detection (0.2$\mu$s). Accordingly, the latency of the assembled 16 DAPs is
    $14\times160\times4/(4.41\times10^9)+(14\times160-4)\times(1024/(1.95\times10^6\times16)+(0.5+0.2)\times10^{-6})=0.0750 s,$
    and the energy consumption is
    $14\times160\times4\times53.96\times10^{-9}\times16+(14\times160-4)\times1024\times178.5\times10^{-9}\times16=6.547J.$
    In comparison, our RRAM-based design is 329.2$\times$ faster and $>10^5$ more energy efficient than this baseline (i.e., 16-DAPs).
    \item[b) ] \emph{Comparison with the combination of isolated digital baseband modules:}
    For the DFT module, we use the high-throughput FFT processor proposed in \cite{liu2018high}. The processor is fabricated using 65nm CMOS technology and is designed to support 5G high-speed requirements. The clock frequency is 250MHz. The processing latency and energy consumption for one 1024-point FFT operation are estimated using the given values of 1080-point FFT (688 clock cycles) and 1200-point FFT (2.07 FFTs/$\mu$J), respectively. 

    For the L-MMSE MIMO detection module, consider the design presented in \cite{tang20192} which was fabricated using the 65nm CMOS technology. The clock frequency is 625MHz. The digital MIMO detector is based on the lower-upper (LU) decomposition which is a two-step algorithm. To estimate the computational workload of solving the related linear equations, the classical Gaussian elimination algorithm is adopted, whose complexity is contributed by two parts: the forward elimination and the backward substitution. Using the terminology, we count the forward substitution step once for channel matrix on all $N_c$ sub-carriers and backward substitution up to $(M-N)$ times for all received data symbols on each sub-carrier. As for the energy consumption, a measurement value of 19.2pJ/b is presented in the literature, which is equivalent to 19.2$\times$8=153.6pJ per MIMO detection. 

    For the combined baseband processing system, the sum processing latency of the two modules with respective 250MHz and 625MHz clocks for DFT and MIMO modules is given by:
    $688\times14\times160/(250\times10^6)+1024\times(24+12\times(14\times160-4))/(625\times10^6)=0.0502 s.$
    In contrast, the latency of our design completing the operations under the same parameters is only 0.2321ms, approximately 220.4$\times$ faster than the digital CMOS counterpart. 

    On the other hand, the energy consumption of the digital CMOS design is estimated as follows:
    $(14\times160/2.07)\times10^{-6}+153.6\times12\times(14\times160-4)\times1024\times10^{-12}=0.0053 J,$
    which is approximately 670.9$\times$ more than our proposed design. It is worth mentioning that the processing latency of the considered digital processor does not meet the stringent requirement of 5G. Traditionally, in order to reduce the processing time to below 10 ms, multiple MIMO detection units should be employed to enable parallel processing, with the sacrifices in chip area and energy consumption.
\end{itemize}
The discussions on comparison with CMOS-based digital processors are summarized in the Table II.

\begin{figure}[t!]
    \centering
    \includegraphics[width=0.48\textwidth]{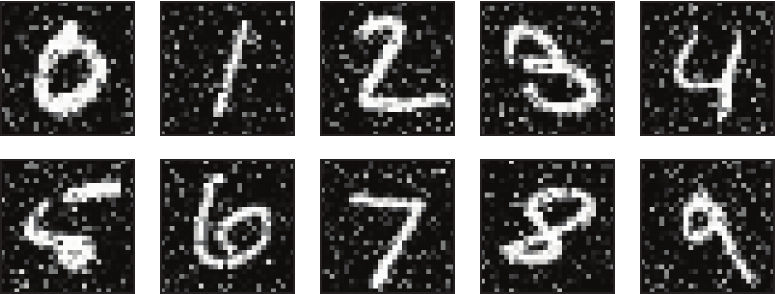}
    \caption{An illustration of noisy test images. In this example, the test images are uncoded transmitted over a noisy channel (SNR = 15dB) and recovered by RRAM-based baseband processing.}\label{fig: noisy images}
\end{figure}

\begin{figure}[t!]
    \centering
    \includegraphics[width=0.48\textwidth]{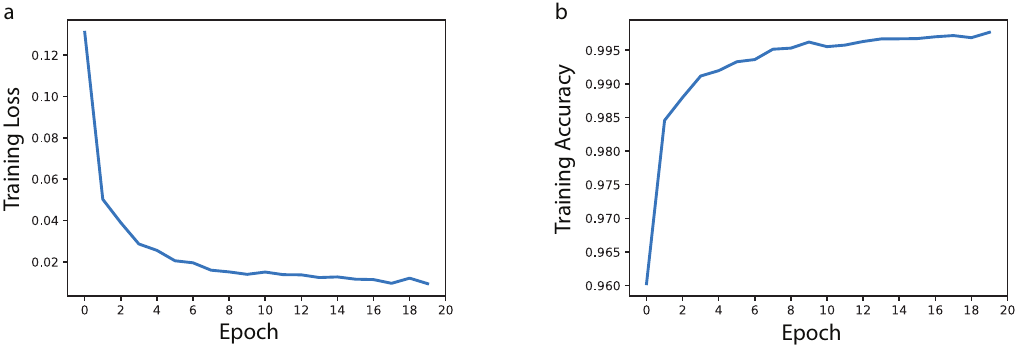}
    \caption{The history of training process. (a) Training loss versus epoch. (b) Training accuracy versus epoch.}
\end{figure}

\begin{table*}[t!]
\centering
\caption{Performance of RRAM-based Baseband Processing for Inference Task}\label{table: accuracy}
\footnotesize
\begin{tabular}{|c|c|c|c|}
\hline
\begin{tabular}[c]{@{}c@{}}Channel Condition\\ SNR (dB)\end{tabular} & \begin{tabular}[c]{@{}c@{}}Digital Baseband Processsing\\ Inference Accuracy (\%)\end{tabular} & \begin{tabular}[c]{@{}c@{}}RRAM-based Baseband Processing\\ Inference Accuracy (\%)\end{tabular} & Performance Loss (\%) \\ \hline
30                                                                   & 98.69                                                                                          & 98.58                                                                                            & 0.11                  \\ \hline
25                                                                   & 98.48                                                                                          & 98.34                                                                                            & 0.14                  \\ \hline
20                                                                   & 97.91                                                                                          & 97.48                                                                                            & 0.43                  \\ \hline
15                                                                   & 94.11                                                                                          & 93.82                                                                                            & 0.29                  \\ \hline
\end{tabular}
\end{table*}

\subsection{Neural Network Application}
In this note, we demonstrate that in-memory baseband processing is more effective for applications with less stringent precision requirements. For example, if the transmitted messages, such as images, are inputs to the downstream neural networks for inference, the models’ robustness against programming noise can ensure high classification accuracy. The simulation is divided into two steps: 1) image transmission; 2) neural network inference. The two procedures are simulated in C++ and Python (TensorFlow), respectively.
\subsubsection{Noisy Images}
Consider the scenario that images are transmitted from sender to destination over a noisy wireless channel. The receiver is equipped with the proposed RRAM-based baseband processor. The images then are fed into a neural network for classification. Due to channel and circuit noise, bit errors occur in the received images. In the simulation, we transmit the test dataset (10,000 images) of MNIST database over noisy channels (characterized by SNRs) using MIMO-ODFM air interface. Then the receiver decodes the images through RRAM-based baseband processing. The parameters follow the standard of 5G NR as well as the behavioral model of our fabricated RRAM device. The examples of resultant noisy test images are shown in Fig. \ref{fig: noisy images}. 
\subsubsection{Neural Network}
The neural network is located at the receiver side to perform classification tasks. The architecture is described as follows: The classifier model is implemented by a 5-layer convolutional neural network (CNN) that consists of two 3$\times$3 convolutional layer, each followed with a 2$\times$2 max pooling, two fully connected layers (the first with 512 units, the second with 64) with ReLU activation, and a final softmax output layer. The CNN classifier is trained using noiseless training dataset (60,000 images) of MNIST database. The objective of training is to minimize the loss function, which is chosen as cross entropy, using the Adam optimizer. The training process in terms of loss and accuracy is illustrated in Fig. 2. After training, the training and testing accuracy achieve up to 99.73\% and 98.74\%, respectively.
\subsubsection{Inference Task}
After transmission and baseband processing, the noisy images (i.e., test dataset) are fed into the trained CNN classifier for inference. We perform simulation under noisy channels with different SNRs. The inference performance is compared to the accuracy from noiseless test dataset, and the performance losses are calculated. The results are listed in Table \ref{table: accuracy}.

\begin{figure}[t!]
    \centering
    \includegraphics[width=0.49\textwidth]{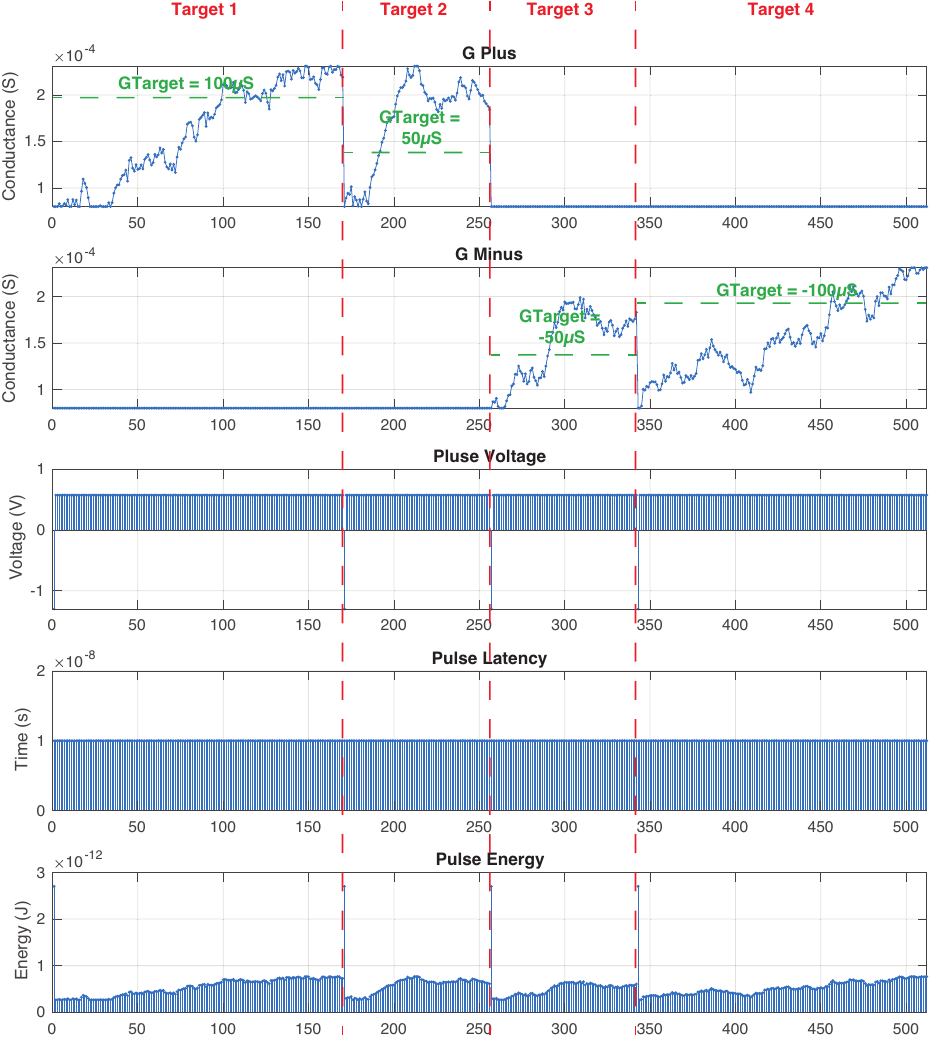}
    \caption{Conductance tuning process of a differential RRAM pair using write-without-verification scheme. (a) The trace of conductance changes of the RRAM device representing the positive value in the differential pair, denoted as G Plus. (b) The trace of conductance changes of the RRAM device representing the negative value in the differential pair, denoted as G Minus. (c) The voltages of write pulses applied to G Plus or G Minus. Each write pulse causes a corresponding conductance change of G Plus or G Minus. (d, e) The latency and energy induced by applying the corresponding write pulses.}\label{fig: 14}
\end{figure}
\begin{figure}[t!]
    \centering
    \includegraphics[width=0.48\textwidth]{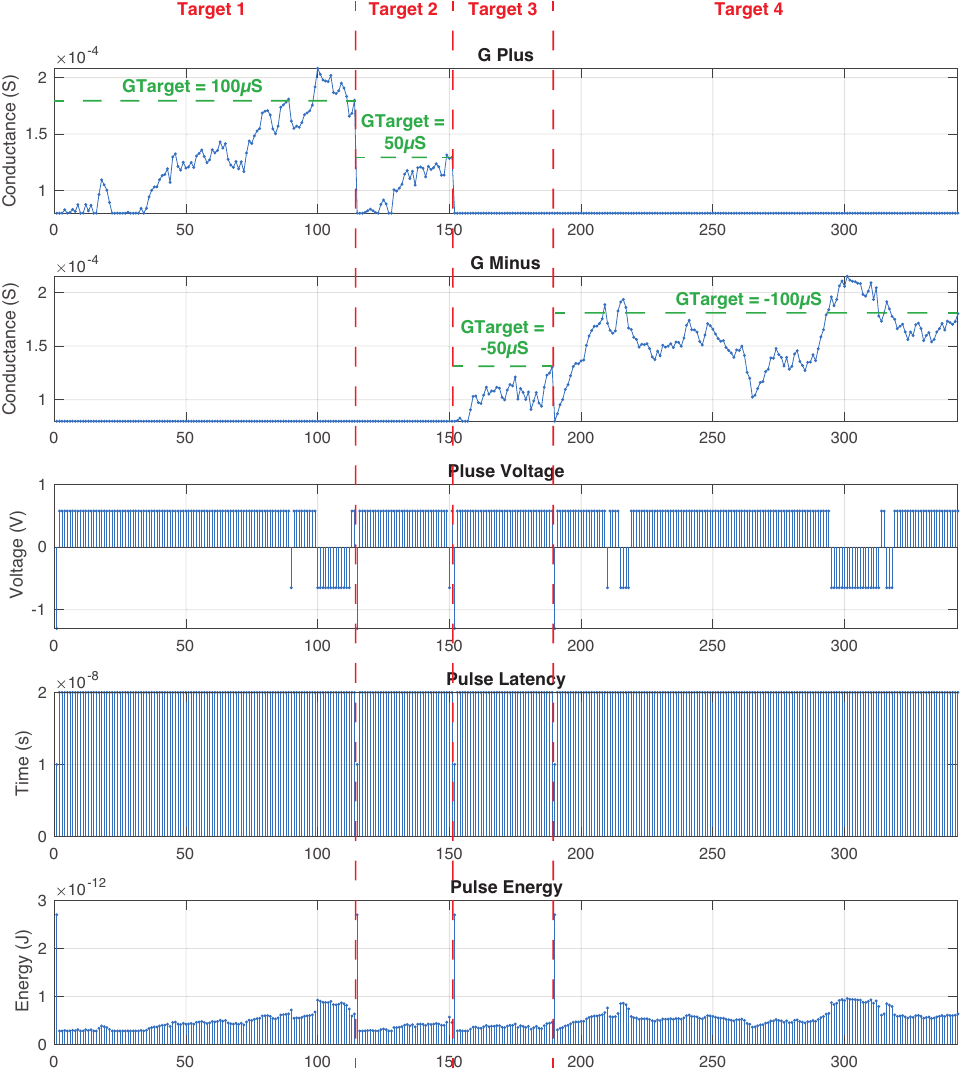}
    \caption{Conductance tuning process of a differential RRAM pair using write-with-verification scheme. (a) The trace of conductance changes of the RRAM device representing the positive value in the differential pair, denoted as G Plus. (b) The trace of conductance changes of the RRAM device representing the negative value in the differential pair, denoted as G Minus. (c) The voltages of write pulses applied to G Plus or G Minus. Each write pulse causes a corresponding conductance change of G Plus or G Minus. (d, e) The latency and energy induced by applying the corresponding write pulses.}\label{fig: 15}
\end{figure}

\subsection{Simulations of Conductance Programming}
In this note, we discuss the modeling of RRAM’s writing process in our simulations. The simulation codes are developed in C++, complied in Clang++-13, and run on Linux system. The simulation is based on the measurement of conductance programming characteristic of our fabricated RRAM device. We considered two schemes: 1) “write-without-verification”, which pre-determines the number and magnitude of write pulses without monitoring the conductance change during the programming process; 2) “write-with-verification”, which used a read pulse after each write pulse to ensure the conductance of the RRAM device can be set to high accuracy. For comparison and clarification, we summarize the two writing schemes as follows:
\begin{itemize}
    \item Write-without-verification scheme: The number of pulses and the voltage direction are determined based on the fabricated RRAM behavioral model shown as the curve presented in Fig. \ref{Fig.2}(g). Then the write pulses are applied to the RRAM device regardless of its intermediate states, i.e., no read operation is applied after each write pulse.
    \item Write-with-verification scheme: A read pulse is applied after each write pulse to get the RRAM device’s conductance. In this way, the state of RRAM is monitored in real-time during the writing process. Based on the read-out conductance value, the controller determines whether the target is met, or additional positive (or negative) pulse should be applied. 
\end{itemize}
We emphasize that the simulation is more authentic than the model in theoretical analysis which is simplified to obtain the theoretical scaling law. The simulation results come from the more accurate modeling of RRAM’s behaviour and complicated programming algorithm. To be specific, the pulses supplied are discrete and thus the conductance change may exceed the threshold. Then the negative pulses will be applied to fine tune the conductance of RRAM. The process may involve multiple stages of positive pulses and negative pulses before convergence to the target value. Moreover, the read noise is included in our simulations ($\sim$1$\mu$S), which also affects the precision of conductance programming and thus deteriorate the performance. We present  Fig. \ref{fig: 14} and  Fig. \ref{fig: 15} to illustrate the evolution of conductance values, supplied voltage pulses, latency, and energy consumption during updating a RRAM differential pair from initial state to conductance targets. For convenience, we name the RRAM device representing the positive value in the differential pair as G Plus, and the other one as G Minus. 
The conductance programming follows the exact procedures described as follows:
\begin{itemize}
    \item Step 1: Initialization. A negative pulse with high voltage (-1.5V) is applied to both G Plus and G Minus, such that the two RRAM devices are fully reset to the minimum conductance state.
    \item Step 2: Tuning. The write pulses are applied to either G Plus or G Minus, depending on the conductance target. If the target is positive, only G Plus is tuned while G Minus remains unchanged during the writing process, vice versa. For write-without-verification scheme, the pre-determined number of pulses are supplied to the RRAM device. For write-with-verification scheme, the iterations between write pulse (0.65V/-0.575V) and read pulse (0.15V) proceed until the conductance of the RRAM device is close enough to the target value within a tolerable error range. 
\end{itemize}
Our simulations are validated by the congruence between the conductance tuning trace depicted in Fig. \ref{fig: 14} and Fig. \ref{fig: 15}, and the corresponding real measurement in the experiments.

\subsection{Mapper and Demapper Modules}
In this note, we present the designs for mapper and demapper modules. The mapper module is a binary to Gray code converter, which converts binary code to Gray code at the beginning of baseband processing at the transmitter side. On the contrary, the demapper module is a Gray to binary code converter, which converts Gray code back to binary code at the end of baseband processing at the receiver side.
\subsubsection{Gray Code}
The Gray code, also known as reflected binary code, is defined as an ordering of the binary numeral system such that the two adjacent values only differ in one bit. Table \ref{table: gray code} shows the relation between decimal numbers, binary and Gray codes.
\begin{table}[t!]
\centering
\caption{The conversion between binary code to Gray code and decimal to Gray code}\label{table: gray code}
\begin{tabular}{|c|c|c|}
\hline
Decimal Numbers & Binary   Code & Gray   Code \\ \hline
0               & 0000          & 0000        \\ \hline
1               & 0001          & 0001        \\ \hline
2               & 0010          & 0011        \\ \hline
3               & 0011          & 0010        \\ \hline
4               & 0100          & 0110        \\ \hline
5               & 0101          & 0111        \\ \hline
6               & 0110          & 0101        \\ \hline
7               & 0111          & 0100        \\ \hline
8               & 1000          & 1100        \\ \hline
9               & 1001          & 1101        \\ \hline
10              & 1010          & 1111        \\ \hline
11              & 1011          & 1110        \\ \hline
12              & 1100          & 1010        \\ \hline
13              & 1101          & 1011        \\ \hline
14              & 1110          & 1001        \\ \hline
15              & 1111          & 1000        \\ \hline
\end{tabular}
\end{table}
\begin{figure}[h]
    \centering
    \includegraphics[width=0.48\textwidth]{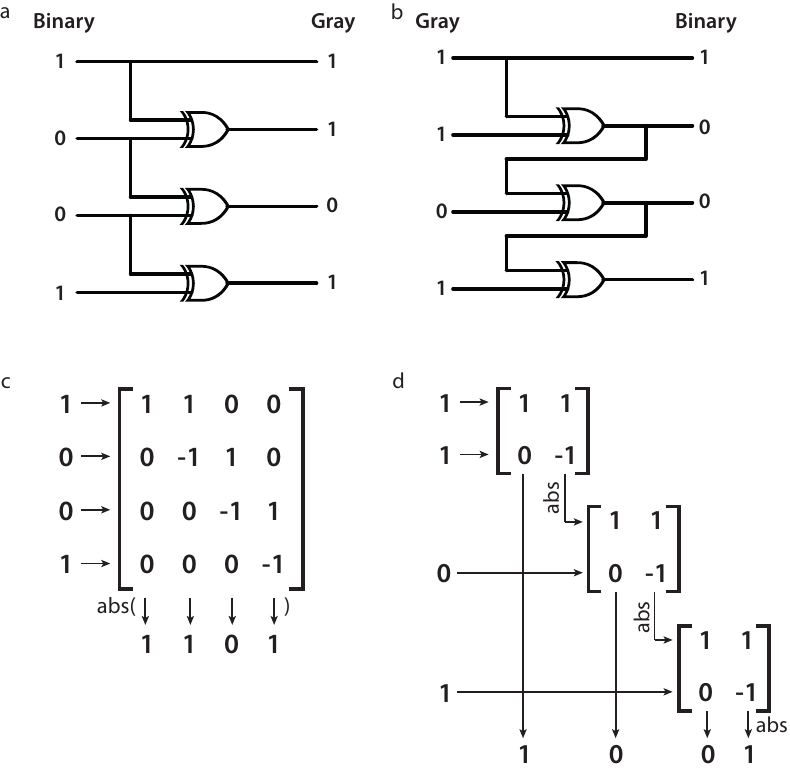}
    \caption{An illustration of the implementations of mapper and demapper. (a) At the transmitter, the logical circuit maps binary bits to Gray codes. After that, they are transformed to analogue values using DAC. (b) At the receiver, the outputs from ADC are Gray coded. The logical circuit maps Gray code back to binary bits. (c) At the transmitter, the codebook as shown in the matrix form is stored in RRAM array. The binary bits are translated as input voltages for the RRAM array. The magnitudes of outputs give the Gray codes. (d) At the receiver, the matrix as shown is stored in three RRAM arrays with the connection between them. The Gray codes are translated as input voltages for them, and the outputs give the binary bits.}\label{fig: 16}
\end{figure}
\subsubsection{Mapper Module}
As for binary to Gray code conversion, several observations can be made from Table \ref{table: gray code} as follows: 1) The first bit of Gray code is equal to the most significant bit of binary code; 2) The second bit of Gray code is the XOR of the first and second bits of the binary code; 3) The third bit of Gray code is the XOR of second and third bits of the binary code; 4) The fourth bit of Gray code is the XOR of third and fourth bits of the binary code. Based on these observations, the logical circuit for binary to Gray code converter (i.e., mapper) is presented in Fig. \ref{fig: 16}(a). To implement the mapper with RRAM array(s), we transform the operations of binary to Gray code converter into matrix-vector multiplication(s). An illustration of the principle of RRAM-based mapper module is shown in  Fig. \ref{fig: 16}(c). The circuit architecture is presented in  Fig. \ref{fig: 17}(a).

\subsubsection{Demapper Module}
As for Gray to binary code conversion, several observations can be made from Table \ref{table: gray code} as follows: 1) The most significant bit of binary code is equal to the first bit of Gray code; 2) If the second bit of Gray code is 0, the second bit of binary code is the same as the previous bit. If the second Gray bit is 1, the second binary bit is the opposite of the previous bit; 3) The operations in 2) continues for the remaining bits. Based on these observations, the logical circuit for Gray to binary code converter (i.e., demapper) is presented in Fig. \ref{fig: 16}(b). To implement the demapper with RRAM array(s), we transform the operations of Gray to binary code converter into matrix-vector multiplication(s). An illustration of the principle of RRAM-based demapper module is shown in Fig. \ref{fig: 16}(d). The circuit architecture is presented in Fig. \ref{fig: 17}(b).

\begin{figure*}[t]
    \centering
    \includegraphics[width=0.75\textwidth]{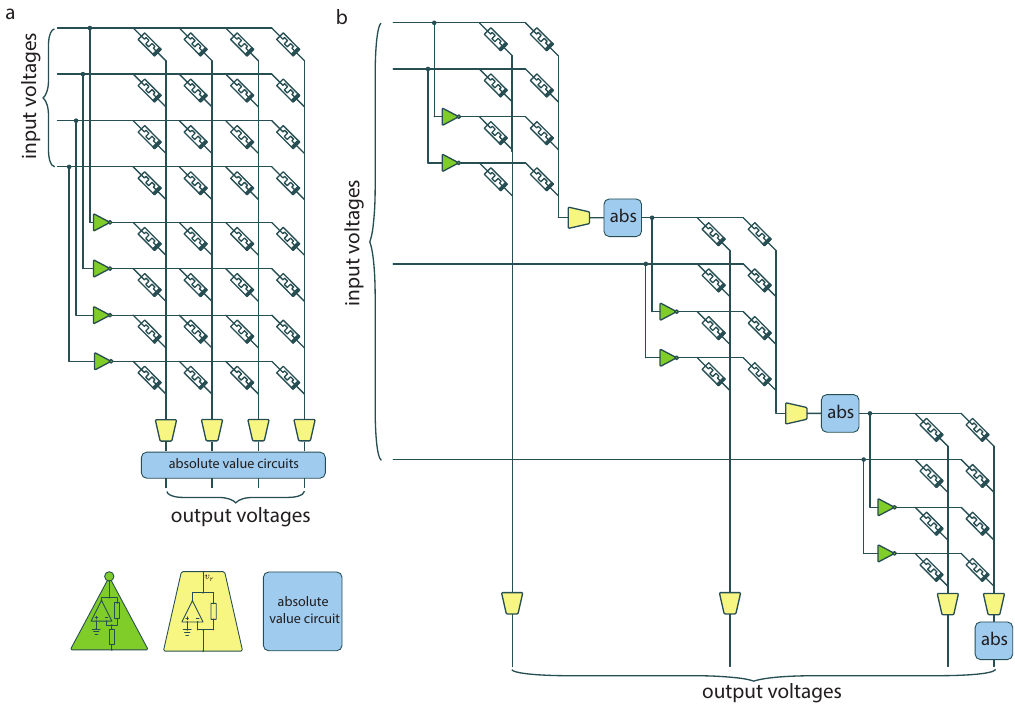}
    \caption{The circuit design of RRAM-based mapper and demapper. (a) The architecture of RRAM-based mapper module. (b) The architecture of RRAM-based demapper module.}\label{fig: 17}
\end{figure*}

\bibliography{reference.bib}
\bibliographystyle{IEEEtran}

\end{document}